\documentclass[12pt]{article}
\usepackage[cp1251]{inputenc}
\usepackage[english]{babel}
\usepackage{amssymb}
\usepackage{amsmath}
\usepackage{amsthm}
\usepackage{amsbsy}
\usepackage{amsfonts}
\usepackage{graphicx}
\usepackage{exscale}
\fontsize{16}{20}
\newtheorem{thm}{Theorem}[section]
\newtheorem{lem}[thm]{Lemma}
\newtheorem{prop}[thm]{Proposition}

\newtheorem{rem}[thm]{Remark}

\numberwithin{equation}{section}

\begin{document}
\begin{center}

\begin{Large}
{\bf Effect of magnetic field on resonant tunneling in 3D waveguides of variable cross-section} 
\end{Large}

\vspace{4ex}

\begin{large}
L.M. Baskin, B.A. Plamenevskii, O.V. Sarafanov
\end{large}

\vspace{6ex}

\end{center}

\begin{abstract}
We consider an infinite three-dimensional waveguide that far from the coordinate origin coincides with a cylinder. The waveguide has two narrows of diameter $\varepsilon$. The  narrows play the role of effective potential barriers for the longitudinal electron motion.
The part of waveguide between the narrows becomes a "resonator"\, and there can arise conditions for electron resonant tunneling. A magnetic field in the resonator can change the basic characteristics of this phenomenon. In the presence of a magnetic field, the tunneling phenomenon  is feasible for producing spin-polarized electron flows consisting of electrons with spins of the same direction.
 
We assume that the whole domain occupied by a  magnetic field is in the resonator. An electron wave function satisfies the Pauli equation in the waveguide and vanishes at its boundary. Taking $\varepsilon$ as a small parameter, we derive asymptotics  for the probability $T(E)$ of an electron with energy $E$ to pass through the resonator, for the "resonant energy"\,$E_{res}$, where $T(E)$ takes its maximal value, and for some other resonant tunneling characteristics.
\end{abstract}

\section{Introduction}\label{s1new}
In this paper, we consider a three-dimensional waveguide that, far from the coordinate origin, coincides with a cylinder $G$ containing the axis $x$. The cross-section of $G$ is a two-dimensional domain (of an arbitrary form) with smooth boundary. The waveguide has two narrows of small diameter $\varepsilon$. The waveguide part between the narrows plays the role of a resonator and there can arise conditions for electron resonant tunneling. This phenomenon consists of the fact that, for  an electron with energy $E$, the probability $T(E)$ to pass from one part of the waveguide to the other through the resonator has a sharp peak at $E=E_{res}$, where $E_{res}$ denotes the "resonant"\, energy. To analyse the operation of devices based on resonant tunneling, it is important to know $E_{res}$, the behavior of $T(E)$ for $E$ close to $E_{res}$, the height of the resonant peak, etc.  

The presence of a magnetic field can essentially affect the basic characteristics of the resonant tunneling
and bring new possibilities for applications in electronics. In particular, in the presence of a magnetic field, the tunneling phenomenon  is feasible for producing spin-polarized electron flows consisting of electrons with spins of the same direction. We suppose that a part of the resonator has been occupied by the magnetic field generated by an infinite solenoid with axis orthogonal to the axis $x$.   Electron wave function satisfies the Pauli equation in the waveguide and vanishes at its boundary (the work function of the waveguide is supposed to be sufficiently large, so such a boundary condition has been justified). Moreover, we assume that only one incoming wave and one outgoing wave can propagate in each cylindrical outlet of the waveguide. In other words, we do not discuss the multichannel electron scattering and consider only electrons with energy between the first and the second thresholds. We take $\varepsilon$ as small parameter and obtain asymptotic formulas for the aforementioned characteristics of the resonant tunneling as $\varepsilon \to 0$. It turns out that such formulas depend on the limiting form of the narrows. We suppose that, in a neighborhood of each narrow, the limiting waveguide coincides with a double cone symmetric about the vertex.

The asymptotic description of electron resonant tunneling in the absence of external fields was presented in \cite{BNPS} for 3D quantum waveguides of similar geometry. Previously there were only episodic studies of the phenomenon by numerical methods, see \cite {LCM}, \cite{BNPP}. \,The extensive literature on the resonant tunneling in 1D waveguides was mainly based on the WKB-method; for our problem the method does not work. In \cite{BNPS}, the  study was based on the compound asymptotic method; the general theory of the method was elaborated in \cite{MNP}. 
In the present paper, we modify the approach in \cite{BNPS} not only analysing the effect of magnetic fields but also developing a more general and simple scheme of study.                     

Section 2 contains statement of the problem. In Section 3, we introduce so-called "limit" boundary value problems, which are independent of the parameter $\varepsilon$. Some model solutions to the problems are studied in Section 4. The solutions will be used in Section 5 to construct asymptotic formulas for appropriate wave functions. In the same section, we investigate the asymptotics of the wave functions and derive asymptotic formulas for main characteristics of the resonant tunneling. Remainders in the asymptotic formulas are estimated in Section 6. 

\section{Statement of the problem}\label{s2new}
To describe the domain  $G(\varepsilon)$ in $\mathbb R^3$ occupied by the waveguide we first introduce domains  $G$ and $\Omega$ in $\mathbb R^3$ independent of $\varepsilon$. The domain $G$ is the cylinder
$$
G=\mathbb R \times D =\{(x, y,z)\in \mathbb R^3: x\in \mathbb
R=(-\infty, +\infty); (y,z) \in D\subset\mathbb{R}^2\}
$$
whose cross-section $D$ is a bounded two-dimensional domain with smooth boundary. Let us \begin{figure}[!htbp]
\centering
\includegraphics[scale=.4]{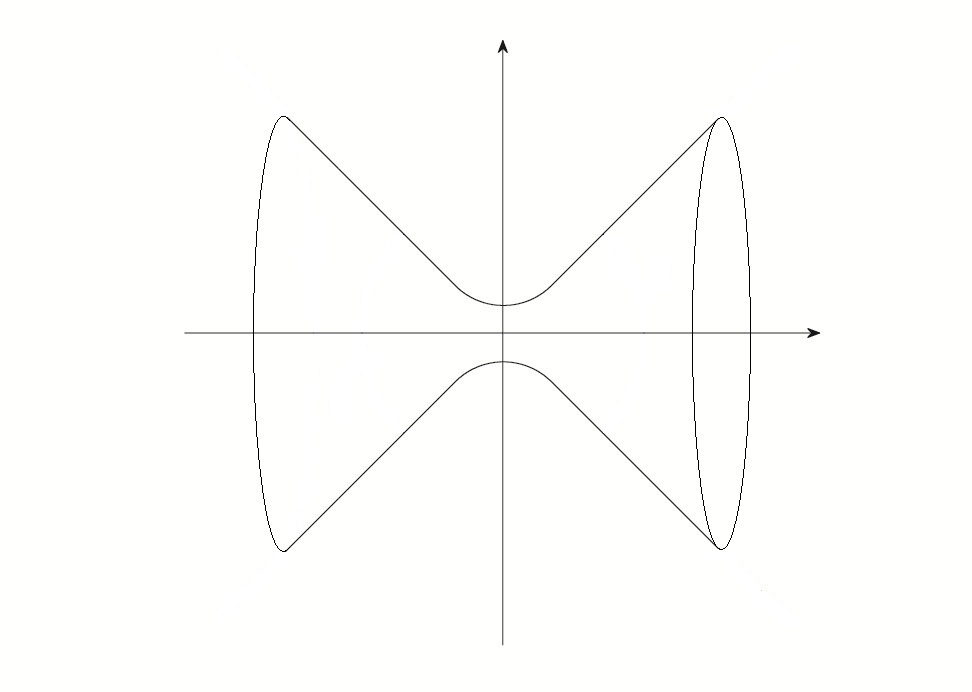}
\caption{The domain $\Omega$.}
\label{Fig.Omega}
\end{figure}
define $\Omega$. Denote by $K$ a double cone with vertex at the coordinate origin $O$ that contains the axis $x$ and is symmetric about the origin. The set $K\cap S^2$ with $S^2$ standing for the unit sphere consists of two
\begin{figure}[h]
\vspace{-1cm}
    \centering
        \includegraphics[scale=.7]{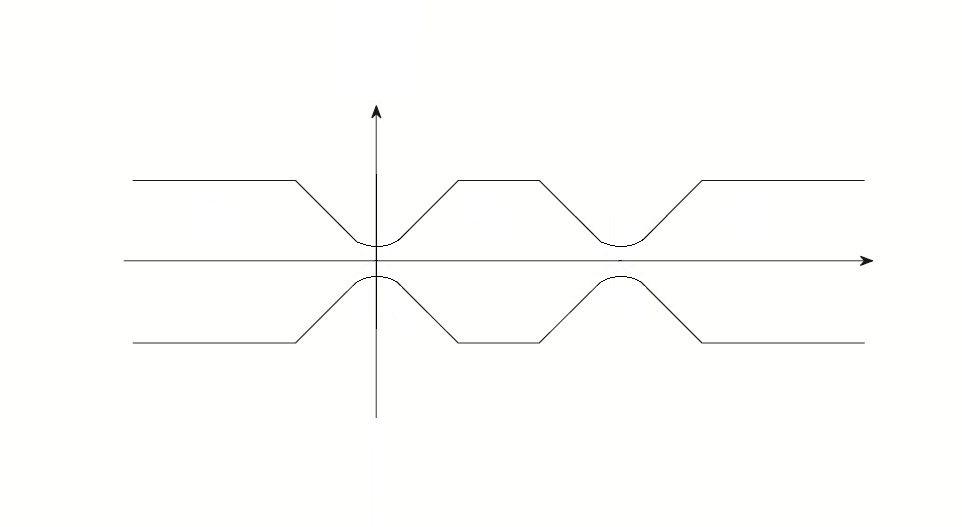}
\vspace{-1cm}
\caption{The waveguide $G(\varepsilon).$}
\end{figure}non-overlapping one-connected domains symmetric about the center of sphere.
Assume that the domain $\Omega$ contains the cone  $K$ together with a neighborhood of its vertex. Moreover,  $\Omega$ coincides with $K$ outside a sufficiently large ball centered at the origin. The boundary $\partial
\Omega$ of  $\Omega$ is supposed to be smooth.

Let us turn to the waveguide $G(\varepsilon)$. We denote by $\Omega
(\varepsilon)$ the domain obtained from $\Omega$ by the contraction with center at $O$ and coefficient  $\varepsilon$. In other words, 
$(x, y,z) \in \Omega (\varepsilon)$ if and only if 
$(x/\varepsilon, y/\varepsilon, z/\varepsilon) \in \Omega$. Let $K_j$ and
$\Omega_j (\varepsilon)$ stand for  $K$ and  $\Omega (\varepsilon)$
shifted by the vector  $\mathbf{r}_j=(x_j^0, 0, 0)$, $j=1, 2$. The value 
$|x_1^0-x_2^0|$ is assumed to be sufficiently large so that the distance 
between $\partial K_1 \cap
\partial K_2$ and  $G$ is positive. We set
$$
G(\varepsilon)=G\cap \Omega_1(\varepsilon)\cap
\Omega_2(\varepsilon).
$$

The wave function ${\bf \Psi}=(\Psi_+,\Psi_-)^T$ of an electron with energy $E=k^2\hbar^2/2m$ in a magnetic field ${\bf H_0}$ satisfies the Pauli equation
\begin{equation}\label{vector Pauli equation}
(-i\nabla+\textbf{A})^2{\bf\Psi}+(\widehat{\bf\sigma},\textbf{H}){\bf\Psi}=k^2{\bf\Psi}\quad  {\text {in}}\quad G(\varepsilon),
\end{equation}
where $\widehat{\bf\sigma}=(\sigma_1,\sigma_2,\sigma_3)$ with the Pauli matrices
$$
\sigma_1=\left(%
\begin{array}{cc}
  1 & 0 \\
  0 & 1 \\
\end{array}%
\right), \quad
\sigma_2=\left(%
\begin{array}{cc}
  0 & -i \\
  i & 0 \\
\end{array}%
\right), \quad
\sigma_3=\left(%
\begin{array}{cc}
  1 & 0 \\
  0 & -1 \\
\end{array}%
\right),
$$
and $\textbf{H}=-(e/c\hbar\textbf{H}_0)=\text{rot}\,\textbf{A}$. If the magnetic field is directed along  the axis $z$ that is 
$\textbf{H}=H\textbf{k}$,  $H$ being a scalar function, then
(\ref{vector Pauli equation}) decomposes into the two scalar equations
\begin{equation}\label{scalar Pauli equation}
(-i\nabla+\textbf{A})^2\Psi_{\pm}\pm H\Psi_{\pm}=k^2\Psi_{\pm}.
\end{equation}
Let the function $H$ depend only on 
$\rho=((x-x_0)^2+(y-y_0)^2)^{1/2}$ with  $H(\rho)=0$ for $\rho>R$,
$R$ being a fixed positive number. Such a field is generated by an infinite 
solenoid with radius  
$R$ and axis parallel to the axis $z$. Then ${\bf A}=A{\bf e}_{\psi}$, where ${\bf
e}_{\psi}=\rho^{-1}(-y+y_0,x-x_0,0)$ and
$$
A(\rho)=\frac{1}{\rho}\left\{%
\begin{array}{ll}
    \int_0^{\rho}tH(t)\,dt, & \rho<R; \\
    \int_0^RtH(t)\,dt, & \rho>R. \\
\end{array}%
\right.
$$
The equality $\hbox{rot}\,\textbf{A}=\textbf{H}$ determines 
$\textbf{A}$ up to a term of the form $\nabla f$. We neglect the waveguide boundary permeability to the electrons and consider the equations  (\ref{scalar Pauli equation}) supplemented by the homogeneous boundary condition
\begin{equation}\label{boundary condition}
    \Psi_{\pm}=0\qquad \hbox{on }\partial G(\varepsilon).
\end{equation}
The obtained boundary value problems are self-adjoint with respect to the Green formulas
\begin{eqnarray*}
((-i\nabla+\textbf{A})^2u\pm Hu-k^2u,v)_{G(\varepsilon)}-
(u,(-i\nabla+\textbf{A})^2v\pm Hv-k^2v)_{G(\varepsilon)}\\+
(u,(-\partial_n-A_n)v)_{\partial
G(\varepsilon)}-((-\partial_n-A_n)u,v)_{\partial G(\varepsilon)}=0,
\end{eqnarray*}
where $A_n$ is the projection of $\textbf{A}$ onto the outward normal to  $\partial
G(\varepsilon)$ and $u,v\in C_c^{\infty}(G(\varepsilon))$ (which means that $u$ and $v$ are smooth functions vanishing outside a bounded set). Besides, 
$\Psi_{\pm}$ must satisfy some radiation conditions at infinity. To formulate such conditions, we have to introduce incoming and outgoing waves. From the requirements on $\textbf{H}$ and the choice of $\textbf{A}$, it can be seen that the coefficients of equations (\ref{scalar Pauli equation}) stabilize at infinity with a power rate. 
Such a slow stabilization offers difficulties in defining these waves. Therefore we will modify $\textbf{A}$  by a gauge transformation so that the coefficients in (\ref{scalar Pauli equation}) become constant for large $|x|$.

Let $(\rho,\psi)$ be polar coordinate on the plane $xy$ centered at
$(x_0,y_0)$ and $\psi =0$ on the ray of the same direction as the axis $x$. We introduce  $f(x,y,z)=c\psi$, where
$c=\int_0^RtH(t)\,dt$. For definiteness, assume that 
$-\pi/2<\psi<3\pi/2$. The function $f$ is uniquely determined in the waveguide for
$|x-x_0|>0$, moreover, $\nabla f=\textbf{A}$ for
$|x-x_0|>R$. Let $\tau$ be a cut-off function on $\mathbb{R}_+$ equal to $1$
for $t>R+2$ and $0$ for $t<R+1$. We set
$\textbf{A}'(x,y,z)=\textbf{A}(x,y,z)-\nabla(\tau(|x-x_0|)
f(x,y,z))$. Then
$\hbox{rot}\,\textbf{A}'=\hbox{rot}\,\textbf{A}=\textbf{H}$ while
$\textbf{A}'=0$ for $|x-x_0|>R+2$. The wave functions 
$\Psi_{\pm}'=\Psi_{\pm}\exp\{i\tau f\}$ satisfy
(\ref{scalar Pauli equation}) with $\textbf{A}$ replaced by
$\textbf{A}'$. For  $|x-x_0|>R+2$, the coefficients of the equations (\ref{scalar
Pauli equation}) with new  vector potential  $\textbf{A}'$ coincide with the coefficients of 
the Helmholtz equation
$$
-\triangle\Psi_{\pm}'=k^2\Psi_{\pm}'.
$$
In order to formulate the radiation conditions, we consider the problem
\begin{eqnarray}\label{problem on D}
  \Delta v(y,z)+\lambda^2v(y,z) &=& 0,\qquad (y,z) \in D,
  \\\nonumber
  v(y,z) &=& 0,\qquad (y,z) \in \partial D.
\end{eqnarray}
The values of parameter $\lambda^2$ that correspond to the nontrivial solutions of this problem form the sequence $\lambda^2_1 < \lambda^2_2 < \dots $ with $\lambda^2_1>0$. These numbers are called the thresholds. Assume that  $k^2$ in (\ref{scalar Pauli equation}) coincides with none of the thresholds and take up the equation in (\ref{scalar Pauli equation})  with $\Psi_+$. For a fixed  $k^2>\lambda_1^2$ there exist finitely many  bounded solutions (wave functions) linearly independent modulo $L_2(G(\varepsilon))$; in other words, a linear combination of such solutions belongs to $L_2(G(\varepsilon))$ if and only if all coefficients are equal to zero. The number of wave functions with such properties remains constant for $k^2 \in (\lambda^2_q, \lambda^2_{q+1})$, $q=1, 2, \dots $ and step-wise increases at the thresholds.

In the present paper, we discuss only the situation, where $k^2 \in (\lambda^2_1, \lambda^2_2)$. In such a case, there exist two independent wave functions. A basis in the space spanned by such functions can be composed of the wave functions  $u^+_1$ and $u^+_2$ satisfying the radiation conditions     
\begin{eqnarray}\label{um}
   u_1^+(x,y,z) &=& \begin{cases}
    e^{i\nu_1x}\Psi_1(y,z) + s_{1 1}^+(k)\,e^{-i\nu_1x}\Psi_1(y,z)+O(e^{\delta x}), & x\rightarrow -\infty, \\
                s_{1 2}^+(k)\,e^{i\nu_1x}\Psi_1(y,z)+O(e^{-\delta x}), & x\rightarrow +\infty; \\
\end{cases} \\\nonumber
  u_{2}^+(x,y) &=& \begin{cases}
    s_{2 1}^+(k)\,e^{-i\nu_1x}\Psi_1(y,z)+O(e^{\delta x}), & x\rightarrow -\infty, \\
    e^{-i\nu_1x}\Psi_1(y,z) +
s_{2 2}^+(k)\,e^{i\nu_1x}\Psi_1(y,z)
+O(e^{-\delta x}), & x\rightarrow +\infty ; \\
\end{cases}
\end{eqnarray}
here $\nu_1=\sqrt{k^2-\lambda^2_1}$ and  $\Psi_1$ stands for an eigenfunction of problem
(\ref{problem on D}) corresponding to
$\lambda^2_1$ and being normalized by the equality
\begin{equation}\label{eigenfunctions}
    \nu_1\int_D|\Psi_1(y,z)|^2\,dy\,dz=1.
\end{equation}
The function  $U_1(x,y,z)=e^{i\nu_1x}\Psi_1(y,z)$ in the cylinder $G$ 
is a wave incoming from $-\infty$ and outgoing to
$+\infty$, while  $U_{2}(x,y,z)=e^{-i\nu_1x}\Psi_1(y,z)$ is a wave going from $+\infty$ to $-\infty$. The matrix 
$$
S^+=\|s_{mj}^+\|_{m,j=1,2}
$$
with entries determined by (\ref{um}) is called the scattering matrix; it is unitary. The quantities 
\begin{equation*}
    R_1^+:=|s_{11}^+|^2, \qquad \,\,\,  T_1^+:=|s_{1 \,2}^+|^2
\end{equation*}
are called the reflection coefficient and the transition coefficient for the wave $U_1$ coming in $G(\varepsilon)$ from $-\infty$. (Similar definitions can be given for the wave $U_{2}$, incoming from  $+\infty$.) In the same manner we introduce the scattering matrix  
$S^-$ and the reflection and transition coefficients $R^-_1$ and $T^-_1$ for the equation in 
(\ref{scalar Pauli equation}) with $\Psi_-$.

We consider only the scattering of the wave going from
$-\infty$ and denote the reflection and transition coefficients by
\begin{equation}\label{2D reflection coeff}
R^\pm=R^\pm(k, \varepsilon)=|s_{11}^\pm(k, \varepsilon)|^2, \qquad \,\,\, T^\pm=T^\pm (k,
\varepsilon)=|s_{12}^\pm(k, \varepsilon)|^2.
\end{equation}
We intend to find a "resonant"\, value  $k_r^\pm=k_r^\pm(\varepsilon)$
of the parameter $k$ which corresponds to the maximum of the transition coefficient and to describe the behavior of $T^\pm(k,\varepsilon)$ near
$k_r^\pm(\varepsilon)$ as $\varepsilon \to 0$.

\section{Limit problems}\label{s3new}

To derive the asymptotics of a wave function (i.e. a solution to problem (\ref{scalar Pauli equation})) as $\varepsilon \rightarrow 0$, we make use of the compound asymptotics method. To this end, we introduce the "limit"\, problems independent of $\varepsilon$. Let the vector potential $\textbf{A}'$ and, in particular, the magnetic field $H$ differ from zero only in the resonator, which is the part of waveguide between the narrows. Then, outside the resonator 
and in a neighborhood of the narrows, the wave function 
under consideration satisfies the Helmholtz equation.

\subsection{ First kind limit problems}
We set $G(0)=G\cap K_1\cap K_2$ (Fig. \ref{Fig.G_0}), so $G(0)$ consists of three parts
$G_1$, $G_2$, and $G_3$. 
\begin{figure}[!htbp]
    \centering
        \includegraphics[scale=.7]{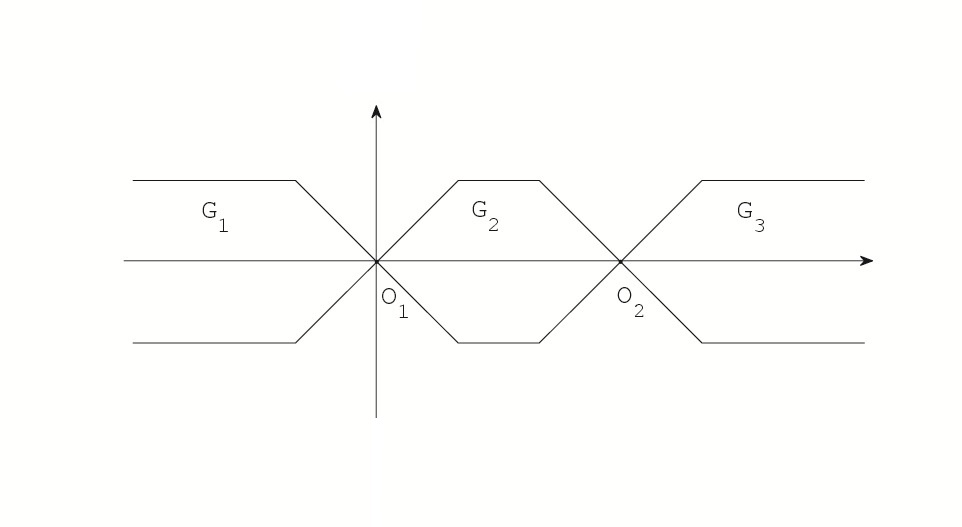}
\vspace{-1cm}
\caption{The domain $G(0).$}
\label{Fig.G_0}
\end{figure}
The boundary value problems
\begin{eqnarray}\label{limit problems kind 1}
  \Delta v(x,y,z)+k^2v(x,y,z) &=& f(x,y,z),\qquad(x,y,z)\in G_j,
  \\\nonumber
  v(x,y,z) &=& 0,\qquad(x,y,z)\in \partial G_j,
\end{eqnarray}
where $j=1,3$, and
\begin{eqnarray}\label{limit problems kind 1 in G2}
  (-i\nabla+\textbf{A}')^2v(x,y,z)\pm H(\rho)v(x,y,z)-k^2v(x,y,z)&=&f(x,y,z),\quad(x,y,z)\in G_2,\qquad
  \\\nonumber
  v(x,y,z) &=& 0,\qquad(x,y,z)\in \partial G_2,
\end{eqnarray}
are called the first kind limit problems.

We introduce function spaces for the problem (\ref{limit problems
kind 1 in G2}) in $G_2$. Denote by $O_1$ and  $O_2$ the conical points of the boundary $\partial G_2$ and  by $\phi_1$ and $\phi_2$ smooth real functions on the closure $\overline{G_2}$ of $G_2$ such that $\phi_j=1$ in a neighborhood of $O_j$ while $\phi_1^2+ \phi_2^2=1$. For $l=0, 1, 2$ and $\gamma \in \mathbb R$, we denote by
$V^l_{\gamma}(G_2)$ the completion in the norm
\begin{equation}\label{2.2}
\|v; V^l_{\gamma}(G_2)\|
=\left(\int_{G_2}\sum_{|\alpha|=0}^{l}\sum_{j=1}^{2}\phi_j^2(x,y,z)r_j^{2(\gamma-l+
|\alpha|)}|\partial^{\alpha}v(x,y,z)|^2\,dx\,dy\,dz \right)^{1/2}
\end{equation}
of the set of smooth functions on $\overline{G_2}$  vanishing near $O_1$ and $O_2$; here $r_j$ is the distance between the points $(x,y,z)$ and $O_j$, $\alpha=(\alpha_1, \alpha_2, \alpha_3)$ is the multiindex, and
$\partial^\alpha =\partial^{|\alpha|}/\partial x^{\alpha_1}
\partial y^{\alpha_2}\partial z^{\alpha_3}$.

Let $K_j$ be the tangent cone to $\partial G_2$ at $O_j$ and $S(K_j)$ the domain that $K_j$ cuts out on the unit sphere centered at $O_j$. We denote by $\mu_1(\mu_1+1)$ and $\mu_2(\mu_2+1)$ the first and second eigenvalues of the Dirichlet problem for the Laplace-Beltrami operator in $S(K_1)$, $0< \mu_1(\mu_1+1) < \mu_2(\mu_2+1)$. Moreover, let $\Phi_1$ stand for an eigenfunction corresponding to $\mu_1(\mu_1+1)$ and normalized by     
$$
    (2\mu_{1}+1)\int_{S(K_1)}|\Phi_{1}(\varphi)|^2d\varphi=1.
$$
The next proposition follows from the general results, e.g. see
\cite[Chapters~2 and 4, \S\S 1--3]{NP} or \cite[v.~1, Chapter~1]{MNP}.
\begin{prop}\label{prop2.1}
Assume that  $|\gamma-1|<\mu_1+1/2$.
Then for  $f \in V^0_{\gamma}(G_2)$ and \vadjust{\vspace{3pt}}
any  $k^2$ except the positive increasing sequence  \/ $\{k^2_p\}_{p=1}^\infty$ of eigenvalues 
$k^2_p \to \infty$, there exists a unique solution  $v \in
V^2_{\gamma}(G_2)$ to the problem \/ {\rm(\ref{limit problems kind 1 in G2})} in
$G_2$. The estimate
\begin{equation}
\label{2.3} \|v; V^2_{\gamma}(G_2)\|\leq c\|f; V^0_{\gamma}(G_2)\|
\end{equation}
holds with a constant  $c$ independent of $f$. If $f$ vanishes in a neighborhood of $O_1$ and $O_2$, then $v$ admits the asymptotics
$$
v(x,y,z)=\left\{%
\begin{array}{ll}
    b_1r_1^{-1/2}\widetilde{J}_{\mu_1+1/2}(kr_1)\Phi_1(\varphi_1)+O\bigl(r_1^{\mu_2+1/2}\bigr), & \hbox{$r_1\rightarrow0$;} \\
    b_2r_2^{-1/2}\widetilde{J}_{\mu_1+1/2}(kr_2)\Phi_1(-\varphi_2)+O\bigl(r_2^{\mu_2+1/2}\bigr), & \hbox{$r_2\rightarrow0$} \\
\end{array}%
\right.
$$
near $O_1$ and $O_2$, where \/ $(r_j,\varphi_j)$ are "polar coordinates" centered at  $O_j$, $r_j>0$ and
and $\varphi_j \in S(K_j)$; $ b_j$ are certain constants; $\widetilde{J}_{\mu}$ denotes the Bessel function multiplied by a constant such that
$\widetilde{J}_{\mu}(kr)= r^{\mu}+o(r^{\mu})$. 

Let $k^2=k^2_0$ be an eigenvalue of problem
\/ {\rm(\ref{limit problems kind 1 in G2})}, then the problem \/ {\rm(\ref{limit problems kind
1 in G2})} is solvable if and only if \/ $(f,v_0)_{G_2}=0$ for any eigenfunction 
$v_0$ corresponding to $k^2_0$. Under such conditions there exists a unique solution   $v$ to problem \/
{\rm(\ref{limit problems kind 1 in G2})} that is orthogonal to all these eigenfunctions and satisfies  
\/ {\rm(\ref{2.3})}.
\end{prop}
We turn to problems (\ref{limit problems kind 1}) for $j=1,3$. Let
$\chi_{0, j}$ and $\chi_{\infty, j}$ be smooth real functions on the closure 
$\overline{G}_j$ of $G_j$ such that
$\chi_{0,j}=1$ in a neighborhood of $O_j$, $\chi_{0, j}=0$
outside a compact set, and  $\chi_{0, j}^2 +
\chi_{\infty, j}^2=1$. We also assume that the support  ${\rm
supp}\chi_{\infty, j}$ is in the cylindrical part of 
$G_j$. For $\gamma \in \mathbb R$, $\delta
>0$, and $l=0, 1, 2$, the space  $V^l_{\gamma,\, \delta}(G_j)$
is the completion in the norm
\begin{equation}
\label{2.4} \|v; V^l_{\gamma,\, \delta}(G_j)\|
=\left(\int_{G_j}\sum_{|\alpha|=0}^{l}\,\bigl(\chi_{0,j}^2r_j^{2(\gamma-l+
|\alpha|)}+ \chi_{\infty, j}^2\exp(2\delta x)\bigr)
|\partial^{\alpha}v|^2\,dx\,dy\,dz\right)^{1/2}
\end{equation}
of the set of functions  with compact support smooth on   $\overline
{G}_j$ and equal to zero in a neighborhood of  $O_j$.

By assumption,  $k^2$ is between the first and second thresholds, so in every domain $G_j$ there is only one outgoing wave; let $U_1^-=U_2$ be the outgoing wave in $G_1$ and $U_2^-=U_1$
that in $G_3$ (the definition of the waves $U_j$ in $G$ see in Section \ref{s2new}). The next proposition follows from Theorem 5.3.5 in \cite{NP}.
\begin{prop}\label{prop2.2}
Let \/ $|\gamma-1|<\mu_1+1/2$ and let the homogeneous problem \/
{\rm(\ref{limit problems kind 1})} (with  $f=0$) have no nontrivial solutions 
in $V^2_{\gamma,\, 0}(G_j)$. Then for any right-hand side  $f\in V^0_{\gamma,\, \delta}(G_j)$ 
there exists a unique solution  $v$ to the problem \/ {\rm(\ref{limit problems kind 1})} that admits the  representation \[ v =u +A_j\chi_{\infty, j}U^-_j,
\]
where $A_j=const$, $u \in V^2_{\gamma,\, \delta}(G_j)$ and  $\delta$
is sufficiently small. Moreover there holds the estimate
\begin{equation}
\label{estimate lim probl 1 Gj}
\|u; V^2_{\gamma,\, \delta}(G_j)\|
+|A_j|\leq c\|f; V^0_{\gamma,\, \delta}(G_j)\|,
\end{equation}
with a constant $c$ independent of $f$. If the function $f$ vanishes in a neighborhood of  $O_j$, then  the solution  $v$ in $G_1$ admits the decomposition 
\[ v(x,y, z)=a_1r_1^{-1/2}\widetilde{J}_{\mu_1+1/2}(kr_1)\Phi_1(-\varphi_1)+O\big(r_1^{\mu_2+1/2}\big),\quad
r_1\rightarrow 0,
\]
and for the solution in $G_3$ there holds
\[
v(x,y)=a_2r_2^{-1/2}\widetilde{J}_{\mu_1+1/2}(kr_2)\Phi_1(\varphi_2)+O\big(r_2^{\mu_2+1/2}\big),\quad
r_2\rightarrow 0,
\]
where $a_j$ are certain constants and $\mu_l$ are the same as in the preceding proposition.
\end{prop}

\subsection{ Second kind limit problems.}

In the domains $\Omega_j$, $j=1, 2$, introduced in Section~\ref{s2new}, we consider
the boundary value problems
\begin{equation} \label{limit problems kind 2}
       \triangle w(\xi_j,\eta_j,\zeta_j) = F(\xi_j,\eta_j,\zeta_j),\quad(\xi_j,\eta_j,\zeta_j)\in\Omega_j; \qquad
               w(\xi_j,\eta_j,\zeta_j)  =0, \quad  (\xi_j,\eta_j,\zeta_j)\in\partial\Omega_j,
\end{equation}
which are called the second kind limit problems; here $(\xi_j,
\eta_j,\zeta_j)$ denote Cartesian coordinates with origin at ~$O_j$.

Let $\rho_j={\rm dist}((\xi_j,\eta_j,\zeta_j), O_j)$ and let $\psi_{0,
j}$, $\psi_{\infty, j}$ be smooth real functions on
$\overline{\Omega}_j$ such that $\psi_{0, j} = 1$ for $\rho_j
<N/2$, $\psi_{0, j} =0$ for $\rho_j>N$, and  $\psi_{0, j}^2+
\psi_{\infty, j}^2=1$ with sufficiently large positive $N$. For $\gamma \in \mathbb R$ and $l=0, 1, 2$, the space
$V^l_{\gamma}(\Omega_j)$ is the completion in the norm
\begin{equation} \label{2.7}
\|v; V^l_{\gamma}(\Omega_j)\|
=\left(\int_{\Omega_j}\sum_{|\alpha|=0}^{l}\,\bigl(\psi_{0,j}^2+
\psi_{\infty, j}^2\rho_j^{2(\gamma-l+ |\alpha|)}\bigr)
|\partial^{\alpha}v|^2\,d\xi_jd\eta_jd\zeta_j\right)^{1/2}
\end{equation}
of the set   $C^\infty_c(\overline{\Omega}_j)$ of smooth functions with compact support in
$\overline {\Omega}_j$. The next proposition is a corollary of  Theorem 4.3.6 in  \cite{NP}.
\begin{prop}\label{prop2.3}
Assume that \/ $|\gamma-1|<\mu_1+1/2$. Then for   $F\in
V^0_{\gamma}(\Omega_j)$ there exists a unique solution  $w \in
V^2_{\gamma}(\Omega_j)$ of the problem \/ {\rm(\ref{limit problems kind 2})}
such that
\begin{equation} \label{estimate lim probl 2}
\|w;
V^2_{\gamma}(\Omega_j)\|\leq c\|F; V^0_{\gamma}(\Omega_j)\|,
\end{equation}
with a constant $c$  independent of  $F$. If $F\in
C^\infty_c(\overline{\Omega}_j)$, then the function $w$ is smooth on
\/ $\overline{\Omega}_j$ and admits the representation
\begin{equation}
\label{asympt w}
w(\xi_j,\eta_j,\zeta_j)=\begin{cases}
    \alpha_j\rho_j^{-\mu_1-1}\Phi_1(-\varphi_j)+O\bigl(\rho_j^{-\mu_2-1}\bigr),  & \xi_j<0,  \\
    \beta_j\rho_j^{-\mu_1-1}\Phi_1(\varphi_j)+O\bigl(\rho_j^{-\mu_2-1}\bigr),  & \xi_j>0,
\end{cases}
\end{equation}
with $\rho_j\rightarrow\infty$; here  $(\rho_j,\varphi_j)$ are polar coordinates on
\/ $\Omega_j$ centered at  $O_j$ while $\mu_l$ and $\Phi_1$ are the same as in Proposition
\/~{\rm\ref{prop2.1}}. The constants 
$\alpha_j$ and  $\beta_j$ are given by
\[
\alpha_j=-(F,w_j^l)_{\Omega},\qquad\beta_j=-(F,w_j^r)_{\Omega},
\]
where  $w_j^l$ and $w_j^r$ are unique solutions to the homogeneous problem \/
{\rm(\ref{limit problems kind 2})} that satisfy, for $\rho_j \to
\infty$, the conditions
\begin{align} \label{asymp wjK}
  w_j^l  = \begin{cases}
    \left(\rho_j^{\mu_1}+\alpha\rho_j^{-\mu_1-1}\right)\Phi_1(-\varphi_j)+O\bigl(\rho_j^{-\mu_2-1}\bigr),  & \xi_j<0;  \\
    \beta\rho_j^{-\mu_1-1}\Phi_1(\varphi_j)+O\bigl(\rho_j^{-\mu_2-1}\bigr),  &
    \xi_j>0;
\end{cases}
\\ \label{asymp wjL}
  w_j^r  = \begin{cases}
    \beta\rho_j^{-\mu_1-1}\Phi_1(-\varphi_j)+O\bigl(\rho_j^{-\mu_2-1}\bigr),  & \xi_j<0; \\
    \left(\rho_j^{\mu_1}+\alpha\rho_j^{-\mu_1-1}\right)\Phi_1(\varphi_j)+O\bigl(\rho_j^{-\mu_2-1}\bigr),\qquad  &
    \xi_j>0.
\end{cases}
\end{align}
The coefficients  $\alpha$ and $\beta$ depend only on the domain $\Omega$.
\end{prop}

\section{Special solutions of limit problems}\label{s4new}

In each domain  $G_j$, $j=1,2,3$, we introduce special solutions to the homogeneous problems (\ref{limit problems kind 1}). Such solutions will be needed in the next section for constructing the asymptotics of a wave function. From Propositions  \ref{prop2.1} and \ref{prop2.2} it follows that the bounded solutions of the homogeneous problems  (\ref{limit problems kind 1}) are trivial (except the eigenfunctions of the problem in $G_2$), so we will consider solutions unbounded in a neighborhoods  of the points $O_j$.

Let us consider the problem in the cone $K_1$, which is, as in Proposition ~\ref{prop2.1}, 
the tangent cone to $\partial G_2$ at $O_1$: 
\begin{equation} \label{problem in K}
  \Delta u+k^2u=0 \,\,\mbox{in}\, K,\qquad u = 0 \,\,\mbox{on}\,\,\partial K.
\end{equation}
The function 
\begin{equation} \label{vK}
v(r,\varphi)= r^{-1/2}\widetilde{N}_{\mu+1/2}(kr)\Phi_1(\varphi)
\end{equation}
satisfies problem (\ref{problem in K}); here  $\widetilde{N}_{\mu}$ is the Neumann function
multiplied by such a constant that \[
\widetilde{N}_{\mu}(kr)= r^{-\mu}+o(r^{-\mu}),\]
and  $\Phi_1$ is the same function as in Proposition ~\ref{prop2.1}. Let 
$t\mapsto\Theta(t)$ be a cut-off function on $\mathbb{R}$ equal to ~1 for
$t<\delta/2$ and ~0 for $t>\delta$ with a small positive $\delta$. We introduce the solution
\begin{equation} \label{V1}
    \textbf{v}_1(x,y,z)=\Theta(r_1)
    v(r_1,\varphi_1)+\widetilde{v}_1(x,y,z)
\end{equation}
to the homogeneous problem  (\ref{limit problems kind 1}) in $G_1$, whereas
$\widetilde{v}_1$ satisfies  (\ref{limit problems kind 1}) with
$f=-[\triangle,\Theta]v$; the existence of $\widetilde{v}_1$
is provided by Proposition ~\ref{prop2.2}. Thus
\begin{equation} \label{asympt v1 at inf}
\textbf{v}_1(x,y,z)=\begin{cases}
     r_1^{-1/2}\bigl(\widetilde{N}_{\mu_1+1/2}(kr_1)+a\widetilde{J}_{\mu_1+1/2}(kr_1)
    \bigr)\Phi_1(-\varphi_1)+O(r_1^{\mu_2}), & r_1\rightarrow0, \\
    AU_1^-(x,y,z)+O(e^{\delta x}),\quad x\rightarrow -\infty, &   \\
\end{cases}
\end{equation}
where  $\widetilde{J}_{\mu}$ is the same function as in 
Proposition ~\ref{prop2.1} and ~\ref{prop2.2} and the constant $A\neq 0$
depends only on the domain $G_1$.

In the domain $G_3$, we introduce the solution $\textbf{v}_3$ to the homogeneous problem 
(\ref{limit problems kind 1}), 
$\textbf{v}_3(x,y,z):=\textbf{v}_1(d-x,-y,-z)$, where
$d=\mbox{dist}(O_1,O_2)$. Then
\begin{equation} \label{asympt v3 at inf}
\textbf{v}_3(x,y,z)=\begin{cases}
     r_2^{-1/2}\bigl(\widetilde{N}_{\mu_1+1/2}(kr_2)+a\widetilde{J}_{\mu_1+1/2}(kr_2)
    \bigr)\Phi_1(\varphi_2)+O(r_2^{\mu_2}), & r_2\rightarrow0, \\
    Ae^{-i\nu_1 d}U_2^-(x,y,z)+O(e^{-\delta x}),\quad x\rightarrow+\infty. &   \\
\end{cases}
\end{equation}

\begin{lem}\label{lem Aa}
There holds the equality\/ $|A|^2={\rm Im}\, a$.\end{lem}

\begin{proof}
Let $(u, v)_Q$ stand for the integral $\int_Q
u\overline{v}\,dx\,dy\,dz$, and let $G_{N,\,\delta}$ be the truncated domain
$G_1\cap\{x>-N\}\cap\{r_1>\delta\}$. By the Green formula
\begin{eqnarray*}
  &&0=(\triangle\textbf{v}_1+k^2\textbf{v}_1,\textbf{v}_1)_{G_{N,\,\delta}}-
(\textbf{v}_1,\triangle\textbf{v}_1+k^2\textbf{v}_1)_{G_{N,\,\delta}}\\&&=
(\partial\textbf{v}_1/\partial n,\textbf{v}_1)_{\partial
G_{N,\,\delta}} -(\textbf{v}_1,\partial \textbf{v}_1/\partial
n)_{\partial G_{N,\,\delta}} = 2i\,{\rm
Im}\,(\partial\textbf{v}_1/\partial n,\textbf{v}_1)_E,
\end{eqnarray*}
where $E= (\partial G_{N,\,\delta}\cap \{x=-N\})\cup (\partial
G_{N,\,\delta}\cap \{r_1=\delta \})$. Taking into account 
(\ref{asympt v1 at inf}) for $x\rightarrow+\infty$ and
(\ref{eigenfunctions}), we have
\begin{eqnarray*} {\rm Im}\,
(\partial\textbf{v}_1/\partial n,\textbf{v}_1)_{\partial
G_{N,\,\delta}\cap \{x=-N\}}=-{\rm Im}\,
\int_DA\dfrac{\partial U_1^-}{\partial
x}(x,y,z)\overline{AU_1^-}(x,y,z)\Big|_{x=-N}dy\,dz+o(1)\\
   = |A|^2\nu_1\int_{-l/2}^{l/2}|\Psi_1(y,z)|^2dy\,dz+o(1)=
   |A|^2+o(1).
\end{eqnarray*}
With (\ref{asympt v1 at inf}) as $r_1\rightarrow 0$ and  the definition of
$\Phi_1$ (see Proposition ~\ref{prop2.1}), we obtain
\begin{align*} {\rm Im}\, (\partial\textbf{v}_1/\partial
n,\textbf{v}_1)_{\partial G_{N,\,\delta}\cap \{r_1=\delta \}} = {\rm
Im}\int_{S(K)}\,\biggl[-\dfrac{\partial}{\partial
r_1}r_1^{-1/2}
\bigl(\widetilde{N}_{\mu_1+1/2}(kr_1)+a\widetilde{J}_{\mu_1+1/2}(kr_1)
    \bigr)\biggr]   \\
     \times
    r_1^{-1/2}\bigl(\widetilde{N}_{\mu_1+1/2}(kr_1)+\overline{a}\widetilde{J}_{\mu_1+1/2}(kr_1)
    \bigr)|\Phi_1(-\varphi_1)|^2r_1^2\Big|_{r_1=\delta}d\varphi_1+o(1)   \\
    = -({\rm Im}\, a)(2\mu_1+1)\int_{S(K)}|\Phi(-\varphi_1)|^2d\varphi_1+o(1)=
   -{\rm Im}\,a+o(1).
\end{align*}
Thus $|A|^2-{\rm Im}\,a+o(1)=0$ as $N\rightarrow\infty$ and
$\delta\rightarrow0$.\qquad\end{proof}

Let  $k^2_{0,\pm}$ be a simple eigenvalue of the problem  (\ref{limit problems kind 1 in G2}) in the resonator  $G_2$ and $v_0^\pm$ is an eigenfunction corresponding to
 $k^2_{0,\pm}$ and normalized by the condition
$\int_{G_2}|v_0^\pm|^2dx\,dy\,dz=1$. By virtue of Proposition ~\ref{prop2.1}
\begin{equation} \label{asympt of
eigenfunction} v_0^\pm(x, y, z)\sim\begin{cases}
    b_1^\pm r_1^{-1/2}\widetilde{J}_{\mu_1+1/2}(k_{0,\pm}r_1)\Phi(\varphi_1), & r_1\rightarrow0, \\
    b_2^\pm r_2^{-1/2}\widetilde{J}_{\mu_1+1/2}(k_{0,\pm}r_2)\Phi(-\varphi_2), & r_2\rightarrow0.  \\
\end{cases}
\end{equation}
We consider that  $b_j^\pm\neq0$. If  $H=0$, then it is true, for instance, for the eigenfunctions corresponding to the minimal eigenvalue of the resonator. For nonzero $H$ this condition can be violated 
owing to the Aharonov-Bohm effect; here we do not discuss this phenomenon. For $k^2$ in a punctured neighborhood of $k^2_{0,\pm}$ separated from the other eigenvalues,
we introduce the solutions  $v_{0j}^\pm$ to the homogeneous problem  (\ref{limit
problems kind 1 in G2}) by the relations 
\begin{equation} \label{vOj}
    v_{0j}^\pm(x,y,z)=\Theta(r_j)
    v(r_j,\varphi_j)+\widetilde{v}_{0j}^\pm(x,y,z),\qquad j=1,2,
\end{equation}
where $v$ is defined by (\ref{vK}) and  $\widetilde{v}_{0j}^\pm$
is a bounded solution to the problem  (\ref{limit problems kind
1 in G2}) with
$f_j(x,y,z)=-[\triangle,\Theta(r_j)]v(r_j,\varphi_j)$.
\begin{lem}\label{lemma3.2}
In a neighborhood  $V\subset \mathbb C$ of $k^2_{0,\pm}$ containing no eigenvalues of the problem \/ {\rm(\ref{limit problems kind 1 in G2})} in 
$G_2$ distinct from $k^2_{0,\pm}$, there hold the equalities
$\widetilde{v}_{0j}^\pm=-\overline{b_j^\pm}(k^2-k^2_{0,\pm})^{-1}v_0^\pm+\widehat{v}_{0j}^\pm$, where
$b_j^\pm$ are the same as in \/ {\rm(\ref{asympt of eigenfunction})} and the functions
$\widehat{v}_{0j}^\pm$ are analytic in $k^2\in
V$.\end{lem}
\begin{proof}
We first verify that  $(v_{0j}^\pm,v_0^\pm)_{G_2}=-\overline{b_j^\pm}/(k^2-k^2_{0,\pm})$, where
$v_{0j}^\pm$ are defined by (\ref{vOj}). We have
\[ (\triangle
v_{0j}^\pm+k^2v_{0j}^\pm,v_0^\pm)_{G_{\delta}}-(v_{0j}^\pm,\triangle
v_0^\pm+k^2v_0^\pm)_{G_{\delta}} =-(k^2-k^2_{0,\pm})(v_{0j}^\pm,v_0^\pm)_{G_{\delta}};
\]
the domain $G_{\delta}$ is obtained from  $G_2$ by cutting out the balls of radius 
$\delta$ with centers at  $O_1$ and $O_2$. Applying the Green formula in the same way as in the proof of Lemma ~\ref{lem Aa}, we arrive at 
$-(k^2-k^2_{0,\pm})(v_{0j}^\pm,v_0^\pm)_{G_{\delta}}=\overline{b_j^\pm}+o(1)$. It remains to let
$\delta \to 0$.

Since  $k^2_{0,\pm}$ is a simple eigenvalue, we have
\begin{equation} \label{vOj B}
\widetilde{v}_{0j}^\pm=\frac{B_j^\pm(k^2)}{k^2-k^2_{0,\pm}}v_0^\pm+\widehat{v}_{0j}^\pm,
\end{equation}
where $B_j^\pm(k^2)$ is independent of  $(x,y,z)$ and $\widehat{v}_{0j}^\pm$ are certain functions analytic in $k^2$ near $k^2=k^2_{0,\pm}$. Multiplying  
(\ref{vOj}) by $v_0^\pm$ and  taking into account  (\ref{vOj B}),
the obtained function for $(v_{0j}^\pm,v_0^\pm)_{G_2}$, and the normalized condition 
$(v_0^\pm,v_0^\pm)_{G_2}=1$, we  arrive at
$B_j^\pm(k^2)=-\overline{b_j^\pm}+(k^2-k^2_{0,\pm})\widetilde{B}_j^\pm(k^2)$, 
$\widetilde{B}_j^\pm$ are being certain analytic functions. Together with 
(\ref{vOj B}), this completes the proof.
\qquad\end{proof}
In view of Lemma ~\ref{lemma3.2}, the expressions   $\mbox{{\bf
v}}_{21}^\pm=(k^2-k^2_{0,\pm})v_{01}^\pm$ and ${\bf v}_{22}^\pm=\overline{b_2^\pm} v_{01}-\overline{b_1^\pm} v_{02}^\pm$
can be extended by continuity to $k^2_{0,\pm}$. According to Proposition~\ref{prop2.1},
\begin{align} 
\label{asympt of v21} \textbf{v}_{21}^\pm(x,y) & \sim
\begin{cases}
    r_1^{-1/2}\bigl((k^2-k^2_{0,\pm})\widetilde{N}_{\mu_1+1/2}(kr_1)+c_1^\pm(k)\widetilde{J}_{\mu_1+1/2}(kr_1)\bigr)\Phi_1(\varphi_1),
    & r_1\rightarrow0,  \\
    c_2^\pm(k)r_2^{-1/2}\widetilde{J}_{\mu_1+1/2}(kr_2)\Phi_1(-\varphi_2), &  r_2\rightarrow0, \\
\end{cases} \\
\label{asympt of v22} \textbf{v}_{22}^\pm(x,y) & \sim
\begin{cases}
    r_1^{-1/2}\bigl(\overline{b_2^\pm}\widetilde{N}_{\mu_1+1/2}(kr_1)+d_1^\pm(k)\widetilde{J}_{\mu_1+1/2}(kr_1)\bigr)\Phi_1(\varphi_1), &  r_1\rightarrow0, \\
    r_2^{-1/2}\bigl(-\overline{b_1^\pm}\widetilde{N}_{\mu_1+1/2}(kr_2)+d_2^\pm(k)\widetilde{J}_{\mu_1+1/2}(kr_2)\bigr)
    \Phi_1(-\varphi_2), &  r_2\rightarrow0. \\
\end{cases}
\end{align}
From the proof of Lemma~\ref{lemma3.2} it follows that
$c_j^\pm(k_{0,\pm})=-\overline{b_1^\pm} b_j^\pm$.

\section{Asymptotic formulas}
In Section \ref{subs. wave function}, we present an asymptotic formula for a wave function (see (\ref{asympt of wave function})), explain its structure, and describe the solutions of the first kind limit problems involved in the formula.  We complete deriving the formula (\ref{asympt of wave function}) in \ref{subs. constants s and C}, where we describe the involved solutions of the second kind limit problems and calculate some coefficients in the expressions for the solutions of the first kind problems.   
In Section \ref{subs. resonant tunnelling}, when analysing the expression for $\widetilde{s}_{12}$ 
obtained in \ref{subs. constants s and C}, we derive formal asymptotics of the resonant tunneling characteristics. Note that the remainders in (\ref{pole}) -- (\ref{quality factor}) have arisen at the intermediate stage of consideration during simplification of the principal part of the asymptotics; they are not the remainders in the final asymptotic formulas. The "final"\, remainders are estimated in the next Section 6, see Theorem \ref{thm6.3}. First, we derive the integral estimate (\ref{estimate R}) of the remainder in (\ref{asympt of wave function}), which proves to be sufficient to obtain more simplified estimates of the remainders in the formulas for the characteristics of resonant tunneling.  The formula  (\ref{asympt of wave function}) and the estimate (\ref{estimate R}) are auxiliary and are analysed only to that extent, which is needed for deriving the asymptotics of tunneling. For ease of notations, we shall in this section drop the symbol  "$\pm$"{}, meaning that we deal with one of the equations (\ref{scalar Pauli equation}).
\subsection{The asymptotics of a wave function}\label{subs. wave function}
In the waveguide $G(\varepsilon)$, we consider the scattering of the wave 
$U(x,y,z)=e^{i\nu_1x}\Psi_1(y,z)$ incoming from $-\infty$ (see
(\ref{eigenfunctions})). The corresponding wave function admits the representation
\begin{eqnarray}\label{asympt of wave function}\nonumber
  &&u(x,y,z;\varepsilon)=\chi_{1,\, \varepsilon}(x,y,z)v_1(x,y,z;\varepsilon)+\\
  &&+\Theta(r_1)w_1(\varepsilon^{-1}x_1,\varepsilon^{-1}y_1,\varepsilon^{-1}z_1;\varepsilon)+
     \chi_{2,\, \varepsilon}(x,y,z)v_2(x,y,z;\varepsilon)+\\\nonumber
  &&+\Theta(r_2)w_2(\varepsilon^{-1}x_2,\varepsilon^{-1}y_2,\varepsilon^{-1}z_2;\varepsilon)+
   \chi_{3,\, \varepsilon}(x,y,z)v_3(x,y,z;\varepsilon)+R(x,y,z;\varepsilon).
\end{eqnarray}
Let us explain the notation and structure of this formula. When constructing the asymptotics, we first describe the behavior of the wave function  $u$ outside the narrows
approximating $u$ by the solutions $v_j$ of the homogeneous problems
(\ref{limit problems kind 1}) and  (\ref{limit problems kind 1 in G2}) in $G_j$. As  $v_j$
we take certain linear combinations of the special solutions introduced in the preceding section; in doing so we subject  $v_1$ and $v_3$ to the same radiation conditions at infinity as $u$:
\begin{eqnarray}\nonumber
    &&v_1 (x,y,z;\varepsilon)=\frac{1}{\overline{A}}\overline{\mathbf{v}}_1(x,y,z)+
    \frac{\widetilde{s}_{11}(\varepsilon)}{A}\mathbf{v}_1(x,y,z)\\\label{v1 radiation condition}
    &&
    \qquad\qquad\qquad\sim
    U_1^+(x,y,z)+\widetilde{s}_{11}(\varepsilon)U_1^-(x,y,z),
    \quad
    \;x\rightarrow-\infty;\qquad\\\label{v2}
    &&v_2(x,y,z;\varepsilon)=C_1(\varepsilon)\mathbf{v}_{21}(x,y,z)+C_2(\varepsilon)\mathbf{v}_{22}(x,y,z);\\
\label{v3 radiation condition}
    &&v_3 (x,y,z;\varepsilon)=\frac{\widetilde{s}_{12}(\varepsilon)}{Ae^{-i\nu_1d}}\mathbf{v}_3(x,y,z)\sim
    \widetilde{s}_{12}(\varepsilon)U_2^-(x,y,z),
    \quad x\rightarrow+\infty;
\end{eqnarray}
for the time being the approximations  $\widetilde{s}_{11}(\varepsilon)$,
$\widetilde{s}_{12}(\varepsilon)$ for the entries   $s_{11}(\varepsilon)$,
$s_{12}(\varepsilon)$ of the scattering matrix and the coefficients  $C_1(\varepsilon)$,
$C_2(\varepsilon)$ are unknown. Here $\chi_{j,\varepsilon}$
stand for  the cut-off functions defined by the equalities
\begin{eqnarray*}
  &&\chi_{1,\, \varepsilon}(x,y,z)=
\left(1-\Theta(r_1/\varepsilon)\right)\mathbf{1}_{G_1}(x,y,z),\qquad
\chi_{3,\, \varepsilon}(x,y,z)=
\left(1-\Theta(r_2/\varepsilon)\right)\mathbf{1}_{G_3}(x,y,z), \\
  &&\chi_{2,\, \varepsilon}(x,y,z)=
\left(1-\Theta(r_1/\varepsilon)-\Theta(r_2/\varepsilon)\right)\mathbf{1}_{G_2}(x,y,z),
\end{eqnarray*}
where  $r_j=\sqrt{x_j^2+y_j^2+z_j^2}$ and $(x_j,y_j,z_j)$ are the coordinates of a point
$(x,y,z)$ in the system with the origin shifted to $O_j$;
$\mathbf{1}_{G_j}$ is the indicator of the set  $G_j$
(equal to 1 in $G_j$ and 0 outside $G_j$); $\Theta(\rho)$ is the same cut-off function as in
(\ref{V1}) (equal to  1 for $0\leqslant\rho\leqslant\delta/2$ and
0 for $\rho\geqslant\delta$ with a fixed sufficiently small positive  $\delta$). Thus  $\chi_{j,\, \varepsilon}$ are defined on the whole waveguide  $G(\varepsilon)$ as well as the functions
$\chi_{j,\, \varepsilon}v_3$ in (\ref{asympt of wave function}).

When substituting  $\sum_{j=1}^3\chi_{j, \, \varepsilon}v_j$ in 
(\ref{scalar Pauli equation}), we obtain the discrepancy in the right-hand side of the Helmholtz equation supported near the narrows. We compensate the principal part of the discrepancy making use of the second kind limit problems. In more detail, we rewrite the discrepancy supported near $O_j$ in the coordinates 
$(\xi_j,\eta_j,\zeta_j)=(\varepsilon^{-1}x_j,\varepsilon^{-1}y_j,
\varepsilon^{-1}z_j)$ in the domain 
$\Omega_j$ and take it as right-hand side for the Laplace equation. Then we rewrite the solution 
$w_j$ of the corresponding problem (\ref{limit problems kind 2})
in the coordinates $(x_2,y_2,z_2)$ and multiply it by the cut-off function. As a result, there arises the term
$\Theta(r_j)w_j(\varepsilon^{-1}x_j,\varepsilon^{-1}y_j,\varepsilon^{-1}z_j;\varepsilon)$
in (\ref{asympt of wave function}).

The existence of solutions $w_j$ vanishing as $O(\rho_j^{-\mu_1-1})$ at infinity follows from Proposition
\ref{prop2.3} (see (\ref{asympt w})). However choosing such solutions and then substituting (\ref{asympt of wave function}) in
(\ref{scalar Pauli equation}), we obtain the discrepancy of high order that has to be compensated again.   Therefore we require 
$w_j=O(\rho_j^{-\mu_2-1})$ as $\rho_j\to\infty$. According to
\ref{prop2.3}, such a solution exists if the right-hand side of the problem
(\ref{limit problems kind 2}) satisfies the additional conditions
\[
(F,w_j^l)_{\Omega_j}=0,\qquad (F,w_j^r)_{\Omega_j}=0.
\]
Such conditions (two at each narrow) uniquely define the coefficients $\widetilde{s}_{11}(\varepsilon)$,
$\widetilde{s}_{12}(\varepsilon)$, $C_1(\varepsilon)$, and $C_2(\varepsilon)$.
The remainder  $R(x,y,z;\varepsilon)$ is small in comparison with the principal part of 
(\ref{asympt of wave function}) as $\varepsilon\rightarrow0$.

\subsection{Formulas for $\widetilde{s}_{11}$, $\widetilde{s}_{12}$, $C_1$, and $C_2$}\label{subs. constants s and C}
We are now going to define the right-hand side $F_j$ of problem (\ref{limit problems kind 2})
and to find $\widetilde{s}_{11}(\varepsilon)$, $\widetilde{s}_{12}(\varepsilon)$,
$C_1(\varepsilon)$, and $C_2(\varepsilon)$. We substitute 
$\chi_{1,\,\varepsilon}v_1$ in (\ref{scalar Pauli equation}) and obtain the discrepancy
$$
(\Delta+k^2)\chi_{1,\,\varepsilon}v_1=[\Delta,\chi_{\varepsilon,1}]v_1+\chi_{\varepsilon,1}(\Delta+k^2)v_1
  = [\Delta,1-\Theta(\varepsilon^{-1}r_1)]v_1,
$$
distinct from zero only near the point $O_1$, where  $v_1$ can be replaced by the asymptotics; the boundary condition (\ref{boundary condition}) is fulfilled.
According to ~(\ref{v1 radiation condition}) and (\ref{asympt v1 at inf}),
$$
v_1(x,y,z;\varepsilon)=r_1^{-1/2}\bigl(a_1^-(\varepsilon)\widetilde{N}_{\mu_1+1/2}(kr_1)
+a_1^+(\varepsilon)\widetilde{J}_{\mu_1+1/2}(kr_1)
    \bigr)\Phi_1(-\varphi_1)+O(r_1^{\mu_2}), \quad r_1\rightarrow0,
$$
with
\begin{equation}\label{a1+-}
a_1^-(\varepsilon)=\frac{1}{\overline{A}}+
    \frac{\widetilde{s}_{11}(\varepsilon)}{A},\quad a_1^+=\frac{\overline{a}}{\overline{A}}+
    \frac{\widetilde{s}_{11}(\varepsilon)a}{A}.
\end{equation}
We single out the principal part of each term and put
$\rho_1=r_1/\varepsilon$, then
\begin{align}\nonumber
  (\Delta+k^2)\chi_{\varepsilon,1}v_1 & \sim
  [\Delta,1-\Theta(\varepsilon^{-1}r_1)]\left(a_1^-r_1^{-\mu_1-1}+a_1^+r_1^{\mu_1}\right)\Phi_1(-\varphi_1)
  \\\label{discrepancy in G1}
   & = \varepsilon^{-2}[\Delta_{(\rho_1,\varphi_1)},1-\Theta(\rho_1)]\left(a_1^-\varepsilon^{-\mu_1-1}\rho_1^{-\mu_1-1}+
   a_1^+\varepsilon^{\mu_1}\rho_1^{\mu_1}\right)\Phi_1(-\varphi_1).
\end{align}
In the same way using (\ref{v2}) and (\ref{asympt of
v21})--(\ref{asympt of v22}), we obtain the principal part of the discrepancy given by
$\chi_{\varepsilon,2}v_2$ supported near  $O_1$:
\begin{equation}\label{discrepancy in G2 near O1}
(\Delta+k^2)\chi_{\varepsilon,1}v_1  \sim
\varepsilon^{-2}[\Delta_{(\rho_1,\varphi_1)},1-\Theta(\rho_1)]\left(b_1^-\varepsilon^{-\mu_1-1}\rho_1^{-\mu_1-1}+
   b_1^+\varepsilon^{\mu_1}\rho_1^{\mu_1}\right)\Phi_1(\varphi_1),
\end{equation}
where
\begin{equation}\label{b1+-}
b_1^-=C_1(\varepsilon)(k^2-k_0^2)+C_2(\varepsilon)\overline{b}_2,\qquad
b_1^+=C_1(\varepsilon)c_1+C_2(\varepsilon)d_1.
\end{equation}
As right-hand side  $F_1$ of the problem  (\ref{limit problems kind 2}) in
$\Omega_1$ we take the function
\begin{align}\nonumber
F_1(\xi_1,\eta_1,\zeta_1)=&-[\Delta,\theta^-]\left(a_1^-\varepsilon^{-\mu_1-1}\rho_1^{-\mu_1-1}+
   a_1^+\varepsilon^{\mu_1}\rho_1^{\mu_1}\right)\Phi_1(-\varphi_1)\\\label{F1}
   &-[\Delta,\theta^+]\left(b_1^-\varepsilon^{-\mu_1-1}\rho_1^{-\mu_1-1}+
   b_1^+\varepsilon^{\mu_1}\rho_1^{\mu_1}\right)\Phi_1(\varphi_1),
\end{align}
where $\theta^+$ (respectively $\theta^-$) stands for the function $1-\Theta$
first restricted to the domain $\xi_1>0$ (respectively $\xi_1<0$) and then extended by 
zero to the whole domain  $\Omega_1$. Let  $w_1$ be the corresponding solution then the term 
$\Theta(r_1)w_1(\varepsilon^{-1}x_1,\varepsilon^{-1}y_1,\varepsilon^{-1}z_1;\varepsilon)$
in (\ref{asympt of wave function}) being substituted in (\ref{scalar Pauli equation}) compensate the discrepancies  (\ref{discrepancy in G1}) --
(\ref{discrepancy in G2 near O1}).

In a similar manner, making use of  (\ref{v2}) -- (\ref{v3 radiation condition}),
(\ref{asympt of v21}) -- (\ref{asympt of v22}), and  (\ref{asympt v3 at
inf}), we find the right-hand side of the problem  (\ref{limit problems kind 2}) for
$j=2$:
\begin{align}\nonumber
F_2(\xi_2,\eta_2,\zeta_2)=&-[\Delta,\theta^-]\left(a_2^-\varepsilon^{-\mu_1-1}\rho_2^{-\mu_1-1}+
   a_2^+\varepsilon^{\mu_1}\rho_2^{\mu_1}\right)\Phi_1(-\varphi_2)\\\nonumber
   &-[\Delta,\theta^+]\left(b_2^-\varepsilon^{-\mu_1-1}\rho_2^{-\mu_1-1}+
   b_2^+\varepsilon^{\mu_1}\rho_2^{\mu_1}\right)\Phi_1(\varphi_2);\\\label{ab2+-}
a_2^-(\varepsilon)=-C_2(\varepsilon)\overline{b}_1,\;\;
a_2^+(\varepsilon)&=C_1(\varepsilon)c_2+C_2(\varepsilon)d_2,\;\;
b_2^-(\varepsilon)=\frac{\widetilde{s}_{12}(\varepsilon)}{Ae^{-i\nu_1 d}},\;\;
b_2^+(\varepsilon)=\frac{a\widetilde{s}_{12}(\varepsilon)}{Ae^{-i\nu_1 d}}.
\end{align}
\begin{lem}\label{lem ab}
If the solution  $w_j$ of the problem  (\ref{limit problems kind 2}) with right-hand side 
\begin{align*}
F_j(\xi_j,\eta_j,\zeta_j)=&-[\Delta,\theta^-]\left(a_j^-\varepsilon^{-\mu_1-1}\rho_j^{-\mu_1-1}+
   a_j^+\varepsilon^{\mu_1}\rho_j^{\mu_1}\right)\Phi_1(-\varphi_j)\\
   &-[\Delta,\theta^+]\left(b_j^-\varepsilon^{-\mu_1-1}\rho_j^{-\mu_1-1}+
   b_j^+\varepsilon^{\mu_1}\rho_j^{\mu_1}\right)\Phi_1(\varphi_j),
\end{align*}
$j=1,2$, admits the estimate 
$O(\rho_j^{-\mu_2-1})$ as $\rho_j\to\infty$, then
\begin{equation}\label{ab conditions}
a_j^-\varepsilon^{-\mu_1-1}-\alpha
a_j^+\varepsilon^{\mu_1}-\beta
b_j^+\varepsilon^{\mu_1}=0,\qquad
b_j^-\varepsilon^{-\mu_1-1}-\alpha
b_j^+\varepsilon^{\mu_1}-\beta a_j^+\varepsilon^{\mu_1}=0,
\end{equation}
where $\alpha$ and $\beta$ are the coefficients in (\ref{asymp
wjK}) -- (\ref{asymp wjL}).
\end{lem}
\begin{proof} By Proposition  \ref{prop2.3},  $w_j=O(\rho_j^{-\mu_2-1})$ as $\rho_j\to\infty$, if and only if the right-hand side of the problem
(\ref{limit problems kind 2}) satisfies the conditions
\begin{equation}\label{F conditions}
(F_j,w_j^l)_{\Omega_j}=0,\qquad (F_j,w_j^r)_{\Omega_j}=0,
\end{equation}
where  $w_j^l$ and  $w_j^r$ are the solutions to the homogeneous problem  (\ref{limit
problems kind 2}) with expansions (\ref{asymp wjK}) --
(\ref{asymp wjL}). We introduce  functions  $f_{\pm}$ on $\Omega_j$
by the equalities
$f_{\pm}(\rho_j,\varphi_j)=\rho_j^{\pm(\mu_1+1/2)-1/2}\Phi_1(\varphi_j)$.
In order to derive (\ref{ab conditions}) from (\ref{F conditions}), 
it suffices to verify that
\begin{align*}
([\Delta,\theta^-]f_-,w_j^l)_{\Omega_j}=([\Delta,\theta^+]f_-,w_j^r)_{\Omega_j}=-1,\quad\,
([\Delta,\theta^-]f_+,w_j^l)_{\Omega_j}=([\Delta,\theta^+]f_+,w_j^r)_{\Omega_j}=\alpha,\\
([\Delta,\theta^+]f_-,w_j^l)_{\Omega_j}=([\Delta,\theta^-]f_-,w_j^r)_{\Omega_j}=0,\qquad
([\Delta,\theta^+]f_+,w_j^l)_{\Omega_j}=([\Delta,\theta^-]f_+,w_j^r)_{\Omega_j}=\beta.
\end{align*}
Let us check the first equalities, the other ones can be considered in a similar way. The support of   $[\Delta,\theta^+]f_-$ is compact, so when calculating 
$([\Delta,\theta^-]f_-,w_j^l)_{\Omega_j}$, one can replace 
$\Omega_j$ by  $\Omega_j^R=\Omega_j\cap\{\rho_j<R\}$ with sufficiently large $R$. Let $E$ denote the set
$\partial\Omega_j^R\cap\{\rho_j=R\}\cap\{\xi_j>0\}$. By the Green formula,
$$
([\Delta,\theta^-]f_-,w_j^l)_{\Omega_j}=
(\Delta\theta^-f_-,w_j^l)_{\Omega_j^R}-(\theta^-f_-,\Delta
w_j^l)_{\Omega_j^R}= (\partial f_-/\partial n,w_j^l)_E-(f_-,\partial
w_j^l/\partial n)_E.
$$
Taking into account  (\ref{asymp wjK}) for $\xi_j<0$ and the definition $\Phi_1$ in 
Proposition ~\ref{prop2.1}, we obtain
\begin{align*}
&([\Delta,\theta^-]f_-,w_j^l)_{\Omega_j}=
\left[\frac{\partial\rho_j^{-\mu_1-1}}{\partial\rho_j}(\rho_j^{\mu_1}+\alpha\rho_j^{-\mu_1-1})-
\rho_j^{-\mu_1-1}\frac{\partial}{\partial\rho_j}
(\rho_j^{\mu_1}+\alpha\rho_j^{-\mu_1-1})\right]
\rho_j^2\Bigg|_{\rho_j=R}\\&\times\int_{S(K)}\Phi(-\varphi_j)^2d\varphi_j+o(1)=
-(2\mu_1+1)\int_{S(K)}\Phi(-\varphi_j)^2d\varphi_j+o(1)=-1+o(1).
\end{align*}
It remains to let  $R\to\infty$. \end{proof}
\begin{rem}
The solutions $w_j$ mentioned in Lemma \ref{lem ab} can be written as linear combinations of certain model
functions independent of $\varepsilon$. We present the corresponding expressions, which will be needed  in the next section for estimating the remainders of asymptotic formulas. Let  $w^l_j$ and $w^r_j$ be the solutions to problem (\ref{limit problems kind 2}) defined by  (\ref{asymp wjK}) -- (\ref{asymp wjL}) and $\theta^+$, $\theta^-$ the same cut-off functions as in (\ref{F1}). We set
\begin{eqnarray*}
\mathbf{w}_j^l=w_j^l- \theta^-\left(\rho_j^{\mu_1}+\alpha\rho_j^{-\mu_1-1}\right)\Phi_1(-\varphi_j) -\theta^+\beta\rho_j^{-\mu_1-1}\Phi_1(\varphi_j),\\
\mathbf{w}_j^r=w_j^r-\theta^-\beta\rho_j^{-\mu_1-1}\Phi_1(-\varphi_j)- \zeta^+\left(\rho_j^{\mu_1}+\alpha\rho_j^{-\mu_1-1}\right)\Phi_1(\varphi_j). 
\end{eqnarray*}
A straightforward verification shows that
\begin{align}\nonumber
w_j&=a_j^+\varepsilon^{\mu_1}\mathbf{w}_j^l+\dfrac{1}{\beta}\left( a_j^-\varepsilon^{-\mu_1-1}-\alpha a_j^+\varepsilon^{\mu_1}\right)\mathbf{w}_j^r\\
\label{wj} 
&=\dfrac{1}{\beta}\left( b_j^-\varepsilon^{-\mu_1-1}-\alpha b_j^+\varepsilon^{\mu_1}\right)\mathbf{w}_j^l+b_j^+\varepsilon^{\mu_1}\mathbf{w}_j^r.
\end{align}
\end{rem}
{\noindent} We use (\ref{a1+-}) and (\ref{b1+-}) to rewrite  (\ref{ab
conditions}) for $j=1$ in the form
\begin{equation}\label{first pair of eq}
\gamma(\varepsilon)\widetilde{s}_{11}(\varepsilon)+\overline{\gamma(\varepsilon)}=C_1(\varepsilon)c_1+C_2(\varepsilon)d_1,\quad
\delta(\varepsilon)\widetilde{s}_{11}(\varepsilon)+\overline{\delta(\varepsilon)}=C_1(\varepsilon)(k^2-k^2_0)+C_2(\varepsilon)\overline{b}_2,
\end{equation}
where
\begin{equation}\label{gammadelta}
    \gamma(\varepsilon)=\frac{1}{A\beta}\left(\varepsilon^{-2\mu_1-1}-a\alpha\right),\qquad
    \delta(\varepsilon)=\frac{1}{A\beta}\left(\alpha+a(\beta^2-\alpha^2)\varepsilon^{2\mu_1+1}\right).
\end{equation}
Moreover, taking account of (\ref{ab2+-}), we rewrite   (\ref{ab conditions}) with $j=2$ in the form
\begin{equation}\label{second pair of eq}
\gamma(\varepsilon)\widetilde{s}_{12}(\varepsilon)=(C_1(\varepsilon)c_2+C_2(\varepsilon)d_2)e^{-i\nu_1d},\quad
\delta(\varepsilon)\widetilde{s}_{12}(\varepsilon)=-C_2(\varepsilon)\overline{b}_1e^{-i\nu_1d}.
\end{equation}
From (\ref{first pair of eq}) and (\ref{second pair of eq}), by means of Lemma \ref{lem Aa}, we obtain $C_1(\varepsilon)$, $C_2(\varepsilon)$,
$\widetilde{s}_{11}(\varepsilon)$, and $\widetilde{s}_{12}(\varepsilon)$:
\begin{align}\label{C1C2}
C_1(\varepsilon)=&(\overline{b}_1c_2)^{-1}\left(\gamma(\varepsilon)\overline{b}_1+\delta(\varepsilon)d_2\right)
\widetilde{s}_{12}(\varepsilon)e^{i\nu_1d},\qquad
C_2(\varepsilon)=-\overline{b}_1^{-1}\delta(\varepsilon)\widetilde{s}_{12}(\varepsilon)e^{i\nu_1d},\\\nonumber
\widetilde{s}_{11}(\varepsilon)
=&(2i\overline{b}_1c_2)^{-1}\bigl((k^2-k^2_0)\overline{b}_1|\gamma(\varepsilon)|^2+
      ((k^2-k^2_0)d_2-\overline{b}_2c_2)\overline{\gamma(\varepsilon)}\delta(\varepsilon)\\\label{s11}
      &-
      \overline{b}_1c_1\gamma(\varepsilon)\overline{\delta(\varepsilon)}-
      (c_1d_2-c_2d_1)|\delta(\varepsilon)|^2\bigr)\widetilde{s}_{12}(\varepsilon)e^{i\nu_1d},
      \\\nonumber
      \widetilde{s}_{12}(\varepsilon) =
      &2i\overline{b}_1c_2e^{-i\nu_1d}\bigl(-(k^2-k^2_0)\overline{b}_1\gamma(\varepsilon)^2-
      ((k^2-k^2_0)d_2-\overline{b}_1c_1-\overline{b}_2c_2)\gamma(\varepsilon)\delta(\varepsilon)\\\label{s12}
      &+
      (c_1d_2-c_2d_1)\delta(\varepsilon)^2\bigr)^{-1}.
\end{align}
\subsection{Asymptotics for resonant tunneling characteristics}\label{subs. resonant tunnelling}
The solutions of the first limit problems involved in  (\ref{asympt of wave function}) are defined for the complex $k^2$ as well. The expression (\ref{s12}) obtained for $\widetilde{s}_{12}$ has a pole at $k^2_p$ in
the lower half-plane. To find  $k^2_p$, we equate 
$2i\overline{b}_1c_2e^{-i\nu_1d}/\widetilde{s}_{12}$ to zero and solve this equation with respect to  $k^2-k^2_0$:
$$
k^2-k^2_0=\left((\overline{b}_1c_1+\overline{b}_2c_2)\gamma(\varepsilon)\delta(\varepsilon)+
      (c_1d_2-c_2d_1)\delta(\varepsilon)^2\right)\left(\overline{b}_1\gamma(\varepsilon)^2+d_2\gamma(\varepsilon)\delta(\varepsilon)\right)^{-1}.
$$
Since the right-hand side of this equation behaves as
$O(\varepsilon^{2\mu_1+1})$ for $\varepsilon\to0$, its solution can be found by the successive approximation method. Taking into account  (\ref{gammadelta}), $c_j(k_0)=-\overline{b}_1b_j$, and Lemma \ref{lem Aa} and neglecting the low order terms, we obtain
 $k^2_p=k^2_r-ik^2_i$,
\begin{equation}\label{pole}
k^2_r=k^2_0-\alpha
(|b_1|^2+|b_2|^2)\varepsilon^{2\mu_1+1}+O(\varepsilon^{4\mu_1+2}),\quad
k_i^2=\beta^2(|b_1|^2+|b_2|^2)|A(k_0^2)|^2\varepsilon^{4\mu_1+2}+O(\varepsilon^{6\mu_1+3}).
\end{equation}
For small  $k^2-k^2_p$,  (\ref{s12}) takes the form
$$
\widetilde{s}_{12}(k,\varepsilon)=-\varepsilon^{4\mu_1+2}\frac{2i\beta^2A(k)^2c_2(k)e^{-i\nu_1d}}{k^2-k^2_p}
\left(1+O(|k^2-k^2_p|+\varepsilon^{2\mu_1+1})\right).
$$
Let $k^2-k^2_0=O(\varepsilon^{2\mu_1+1})$, then
$|k^2-k^2_p|=O(\varepsilon^{2\mu_1+1})$,
$A(k)=A(k^2_0)+O(\varepsilon^{2\mu_1+1})$,
$c_2(k^2)=-\overline{b}_1b_2+O(\varepsilon^{2\mu_1+1})$,
$\nu_1(k)=\nu_1(k^2_0)+O(\varepsilon^{2\mu_1+1})$, and
\begin{align*}
\widetilde{s}_{12}(k,\varepsilon)=\varepsilon^{4\mu_1+2}\frac{2i\beta^2\overline{b}_1b_2A(k_0)^2e^{-i\nu_1(k_0)d}}{k^2-k^2_p}
\left(1+O(\varepsilon^{2\mu_1+1})\right)\\
=\frac{\dfrac{\overline{b}_1}{|b_1|}\dfrac{b_2}{|b_2|}\left(\dfrac{A(k_0)}{|A(k_0)|}\right)^2e^{-i\nu_1(k_0)d}}{\dfrac{1}{2}\left(\dfrac{|b_1|}{|b_2|}+\dfrac{|b_2|}{|b_1|}\right)-iP\dfrac{k^2-k^2_r}{\varepsilon^{4\mu_1+2}}}
\left(1+O(\varepsilon^{2\mu_1+1})\right),
\end{align*}
where $P=(2|b_1||b_2|\beta^2|A(k_0)|^2)^{-1}$. Thus
\begin{equation}\label{T}
\widetilde{T}(k,\varepsilon)=|\widetilde{s}_{12}|^2=\frac{1}{\dfrac{1}{4}\left(\dfrac{|b_1|}{|b_2|}+\dfrac{|b_2|}{|b_1|}\right)^2+P^2\left(\dfrac{k^2-k^2_r}{\varepsilon^{4\mu_1+2}}\right)^2}(1+O(\varepsilon^{2\mu_1+1})).
\end{equation}
The obtained approximation  $\widetilde{T}$ for the transition coefficient has a peak at $k^2=k^2_r$ whose width at its half-height is equal to
\begin{equation}\label{quality factor}
\widetilde{\Upsilon}(\varepsilon)=\left(\dfrac{|b_1|}{|b_2|}+\dfrac{|b_2|}{|b_1|}\right)P^{-1}\varepsilon^{4\mu_1+2}.
\end{equation}

\section{Justification of the asymptotics}

As in the preceding section, here we drop the symbol "$\pm$"{} in notations and do not mention which of the two equations in (\ref{scalar Pauli equation}) is under consideration. We will return to the detailed notation in the formulation of  Theorem \ref{thm6.3}.

We introduce the function spaces for the problem
\begin{equation}
\label{f0 problem in G(eps1eps2)}
 (-i\nabla+\textbf{A})^2u\pm Hu=k^2u \quad \mbox{in}\;G(\varepsilon),\qquad
    u=0\quad \mbox{on}\;\partial G(\varepsilon).
\end{equation}
Recall that the functions  $\textbf{A}$ and $H$ are compactly supported and differ from zero only in the resonator at a distance from the narrows. Let  $\Theta$ be the same function as in  (\ref{V1}). We assume that the cut-off functions $\eta_j$
$j=1,2,3$, are distinct from zero only in $G_j$ and satisfy 
$\eta_1(x,y,z)+\Theta(r_1)+\eta_2(x,y,z)+\Theta(r_2)+\eta_3(x,y,z)=1$ in
$G(\varepsilon)$. With $\gamma \in \mathbb
R$, $\delta
>0$, and  $l=0, 1, 2$ the space $V^l_{\gamma,\delta}(G(\varepsilon))$ is the completion in the norm
\begin{multline}\label{norm G(eps1eps2)}
\|u;
V^l_{\gamma,\delta}(G(\varepsilon))\| \\
=\Biggl(\int_{G(\varepsilon)}\sum_{|\alpha|=0}^{l}\,\Biggl(\sum_{j=1}^2\Theta^2(r_j)\;(r_j^2+\varepsilon_j^2)^{\gamma-l+
|\alpha|}
+\eta_1^2e^{2\delta |x|}+\eta_2+\eta_3^2e^{2\delta |x|}\Biggr)
|\partial^{\alpha}v|^2\,dx\,dy\,dz\Biggr)^{1/2}
\end{multline}
of the set of smooth functions on  $\overline{G(\varepsilon)}$ with compact supports. Denote by
$V_{\gamma,\delta}^{0,\perp}$
the space of functions $f$ that are analytic in $k^2$, take values in  
$V_{\gamma,\delta}^0(G(\varepsilon))$, and, at $k^2=k_0^2$, satisfy 
$(\chi_{2,\varepsilon^{\sigma}}f,v_0)_{G_2}=0$
with a small $\sigma >0$.
\begin{prop}\label{prop5.3}
Assume that $k_r^2$ is a resonant energy,  $k^2_r\to k^2_0$ as $\varepsilon \to 0$, and \/  $|k^2-\nobreak k_r^2|=O(\varepsilon^{2\mu_1+1})$.
We also suppose that  $\gamma$ satisfies
$\mu_1-3/2<\gamma-1<\mu_1+1/2$, $f\in
V^{0,\perp}_{\gamma,\delta}(G(\varepsilon))$,
and $u$ a solution to problem \/ {\rm(\ref{f0 problem in G(eps1eps2)})} that admits the representation 
\[
u= \widetilde{u}+\eta_1A^-_1U_1^-+\eta_3A^-_2U_2^-;
\]
here $A^-_j=const$ and $\widetilde{u}\in
V^2_{\gamma,\delta}(G(\varepsilon))$ with small $\delta>0$. 
Then
\begin{equation}
\label{estimate probl G(eps1eps2)}
 \|\widetilde{u};V^2_{\gamma,\delta}(G(\varepsilon))\|+|A_1^-|+|A_2^-|\leq
c\|f;V^0_{\gamma,\delta}(G(\varepsilon))\|,
\end{equation}
where $c$ is a constant independent of $f$ and $\varepsilon$.\end{prop}
\begin{proof}
{\it Step} A\@. We first construct an auxiliary function $u_p$. As was mentioned,  $\widetilde{s}_{12}$ has the pole $k^2_p=k^2_r-ik^2_i$ (see (\ref{pole})). 
Let us multiply the solutions of limit problems involved in (\ref{asympt of wave function}),
by $A(k)b_2\beta\varepsilon^{2\mu_1+1}/ s_{12}(\varepsilon,k)e^{i\nu_1d}$, set $k=k_p$, and re-denote the obtained functions endowing them with the index $p$.Then 
\begin{align}\label{vp}
   v_{1p}(x,y,z;\varepsilon) &= \varepsilon^{2\mu_1+1}(b_1\beta+O(\varepsilon^{2\mu_1+1}))\textbf{v}_1(x,y,z;k_p),\\\nonumber
   v_{3p}(x,y,z;\varepsilon) &= \varepsilon^{2\mu_1+1}b_2\beta\textbf{v}_1(x,y,z;k_p); \\\nonumber
  v_{2p}(x,y,z;\varepsilon) &= \left(-\frac{1}{\overline{b}_1}+O\big(\varepsilon^{2\mu_1+1}\big)\right)
  \textbf{v}_{21}(x,y,z;k_p) \\\nonumber
  & + \varepsilon^{2\mu_1+1}\left(-\alpha\frac{ b_2}{\overline{b}_1}+O\big(\varepsilon^{2\mu_1+1}\big)\right)
  \textbf{v}_{22}(x,y,z;k_p),\\\nonumber
  w_{1p}(\xi_1,\eta_1,\zeta_1;\varepsilon) &=
  b_1\varepsilon^{2\mu_1+1}\left(\varepsilon^{2\mu_1+1}\left(a(k_p)\beta +O(\varepsilon^{2\mu_1+1})\right)
  \mathbf{w}_1^l(\xi_1,\eta_1,\zeta_1)\right.\\\label{w1p}
  &\left.+\left(1+O(\varepsilon^{2\mu_1+1})\right)\mathbf{w}_1^r(\xi_1,\eta_1,\zeta_1)\right),\\\nonumber
  w_{2p}(\xi_2,\eta_2,\zeta_2;\varepsilon) &=
  b_2\varepsilon^{2\mu_1+1}\left(\left(1+O(\varepsilon^{2\mu_1+1})\right)\mathbf{w}_1^l(\xi_2,\eta_2, \zeta_2)\right.\\\label{w2p}
  &\left.+a(k_p)\beta\varepsilon^{2\mu_1+1}
  \mathbf{w}_1^r(\xi_2,\eta_2,\zeta_2)\right);
\end{align}
the dependence of $k_p$ on $\varepsilon$ has not been indicated. We set
\begin{multline}\label{up}
    u_p(x,y,z;\varepsilon)=\Xi(x,y,z)\left[\chi_{1,\varepsilon}(x,y,z)v_{1p}(x,y,z;\varepsilon)+
    \Theta(\varepsilon^{-2\sigma}r_1)w_{1p}(\varepsilon^{-1}x_1,\varepsilon^{-1}y_1,\varepsilon^{-1}z_1;\varepsilon)\right.\\
    +\chi_{2,\varepsilon}(x,y,z)v_{2p}(x,y,z;\varepsilon)+
    \Theta(\varepsilon^{-2\sigma}r_2)w_{2p}(\varepsilon^{-1}x_2,\varepsilon^{-1}y_2,\varepsilon^{-1}z_2;k,\varepsilon)\\
    +\left.
    \chi_{3,\varepsilon}(x,y,z)v_{2p}(x,y,z;k,\varepsilon)\right],
\end{multline}
where $\Xi$ is a cut-off function on
$G(\varepsilon)$ equal to 1 on 
$G(\varepsilon)\cap\{|x|<R\}$ and  0~on
$G(\varepsilon)\cap\{|x|>R+1\}$ with sufficiently large $R>0$, $(x_j,y_j,z_j)$ are the coordinates of a point $(x,y,z)$ in the system with origin shifted to $O_j$.
The term $\chi_{2,\varepsilon}v_{2p}$ gives the main contribution in the norm of  $u_p$.
In view of the definitions of  $v_{2p}$ and $\mathbf{v}_{21}$ (see Section \ref{s4new}) and Lemma \ref{lemma3.2}, we obtain
$\|\chi_{2,\varepsilon}v_{2p}\|=\|v_0\|+o(1)$. 

{\it Step} B\@. We show that \begin{equation}
\label{estimate vp}
\|((-i\nabla+\textbf{A})^2\pm H-k^2_p)u_p;V^0_{\gamma,\,
\delta}(G(\varepsilon))\|\leq
c\varepsilon^{\mu_1+\kappa},
\end{equation}
where
$\kappa=\min\{\mu_1+1,$~$\mu_2+1-\sigma_1,$~$\gamma+3/2\}$,
$\sigma_1=2\sigma(\mu_2-\gamma+3/2)$. If
$\mu_1-3/2<\gamma-1$ and $\sigma$ is sufficiently small so that
$\mu_2-\mu_1>\sigma_1$, then $\kappa=\mu_1+1$.

By virtue of (\ref{up})
\begin{align*}
   &((-i\nabla+\textbf{A})^2\pm H-k^2_p)u_p(x,y,z;\varepsilon) \\
   & =[\triangle,\chi_{1,\varepsilon}]\left(v_1(x,y,z;\varepsilon)-b_1\beta\varepsilon^{2\mu_1+1}
   (r_1^{-\mu_1-1}+a(k_p)r_1^{\mu_1})\Phi_1(-\varphi_1)\right)\nonumber\\
    &    +[\triangle,\Theta]w_{1p}(\varepsilon^{-1}x_1,\varepsilon^{-1}y_1,\varepsilon^{-1}z_1;\varepsilon)-
   k^2\Theta(\varepsilon^{-2\sigma}r_1)w_{1p}(\varepsilon^{-1}x_1,\varepsilon^{-1}y_1, \varepsilon^{-1}z_1;\varepsilon)\nonumber\\
    &  +[\triangle,\chi_{2,\varepsilon}]\left(v_2(x,y,z;\varepsilon)-\Theta(r_1)\bigl(b_{1p}^-(\varepsilon)
   r_1^{-\mu_1-1}+b_{1p}^+(\varepsilon)r_1^{\mu_1}\bigr)\Phi_1(-\varphi_1)\right.
    \nonumber\\
    & \qquad\qquad\qquad\qquad\quad -\left. \Theta(r_2)\bigl(a_{2p}^-(\varepsilon)
   r_2^{-\mu_1-1}+a_{2p}^+(\varepsilon)r_2^{\mu_1}\bigr)\Phi_1(\varphi_2)\right)
    \nonumber\\
    &   +[\triangle,\Theta]w_{2p}(\varepsilon^{-1}x_2,\varepsilon^{-1}y_2,\varepsilon^{-1}z_2;\varepsilon)-
   k^2\Theta(\varepsilon^{-2\sigma}r_2)w_{2p}(\varepsilon^{-1}x_2,\varepsilon^{-1}y_2, \varepsilon^{-1}z_2;\varepsilon)\nonumber\\
    &  
    +[\triangle,\chi_{3,\varepsilon}]\left(v_3(x,y,z;\varepsilon)-b_2\beta\varepsilon^{2\mu_1+1}
   (r_2^{-\mu_1-1}+a(k_p)r_2^{\mu_1})\Phi_1(\varphi_2)\right)\\
   &  \qquad\qquad \qquad\qquad \qquad\qquad +[\triangle,\Xi]v_1(x,y,z;\varepsilon)+[\triangle,\Xi]v_3(x,y,z;\varepsilon),
\end{align*}
where $b_{1p}^-=O(\varepsilon^{2\mu_1+1})$, $b_{1p}^+=b_1+O(\varepsilon^{2\mu_1+1})$, $a_{2p}^-=O(\varepsilon^{2\mu_1+1})$, $a_{2p}^+=b_2+O(\varepsilon^{2\mu_1+1})$. Taking account of the asymptotics  $\mathbf{v}_1$ as $r_1\rightarrow 0$ and going to  the variables $(\xi_1,\eta_1,\zeta_1)=(\varepsilon^{-1}x_1,\varepsilon^{-1}y_1,\varepsilon^{-1}z_1)$, we arrive at 
\begin{align*}
   \left\|(x,y,z)\mapsto[\triangle,\chi_{1,\varepsilon}]\left(\mathbf{v}_1(x,y,z)-
   (r_1^{-\mu_1-1}+a(k_p)r_1^{\mu_1})\Phi_1(-\varphi_1)\right);
   V^0_{\gamma,\delta}(G(\varepsilon))\right\|^2\\
   \leq c\int_{G(\varepsilon)}(r_1^2+\varepsilon^2)^{\gamma}
   \left|[\triangle,\chi_{1,\varepsilon}] r_1^{-\mu_1+1}\Phi(-\varphi_1)\right|^2dx\,dy\,dz
   \leq c\varepsilon^{2(\gamma-\mu_1+1/2)}.
\end{align*}
This and (\ref{vp}) imply that
\[\left\|(x,y,z)\mapsto[\triangle,\chi_{1,\varepsilon}]\left(v_1(x,y,z)-
   (r_1^{-\mu_1-1}+a(k_p)r_1^{\mu_1})\Phi(-\varphi_1)\right);
   V^0_{\gamma,\delta}(G(\varepsilon))\right\|
   \leq c\varepsilon^{\gamma+\mu_1+3/2}.
\]
Similarly,
\begin{align*}
&\left\|(x,y,z)\mapsto[\triangle,\chi_{2,\varepsilon}]\left(v_2(x,y,z)
   -\Theta(r_1)\bigl(b_{1p}^-(\varepsilon)
   r_1^{-\mu_1-1}+b_{1p}^+(\varepsilon)r_1^{\mu_1}\bigr)\Phi_1(-\varphi_1)\right.\right.
    \\ &\qquad\qquad\qquad\qquad\quad -\left.\left. \Theta(r_2)\bigl(a_{2p}^-(\varepsilon)
   r_2^{-\mu_1-1}+a_{2p}^+(\varepsilon)r_2^{\mu_1}\bigr)\Phi_1(\varphi_2)\right)\right\|
   \leq c\varepsilon^{\gamma+\mu_1+3/2},\\
   &\left\|(x,y,z)\mapsto[\triangle,\chi_{3,\varepsilon}]\left(v_3(x,y,z)-
   (r_2^{-\mu_1-1}+a(k_p)r_2^{\mu_1})\Phi_1(\varphi_2)\right);
   V^0_{\gamma,\delta}(G(\varepsilon))\right\|
   \leq c\varepsilon^{\gamma+\mu_1+3/2}.
\end{align*}
It is clear that
\[\left\|[\triangle,\Xi]v_l;
   V^0_{\gamma,\delta}(G(\varepsilon))\right\|
   \leq c\varepsilon^{2\mu_1+1},\quad l=1,3.\]
Further, since
$\textbf{w}^l_j$ behaves as $O(\rho_j^{-\mu_2-1})$ at infinity, we have
\begin{multline*}
    \int_{G(\varepsilon)}(r^2_j+\varepsilon^2)^{\gamma}\left|[\triangle,\Theta]
    \mathbf{w}^l_j(\varepsilon^{-1}x_j,\varepsilon^{-1}y_j,\varepsilon^{-1}z_j)\right|^2dx_jdy_jdz_j
    \\
    \leq c\int_{K_j}(r_j^2+\varepsilon^2)^{\gamma}\left|[\triangle,\Theta](\varepsilon^{-1}r_j)^{-\mu_2-1}\Phi_2(\varphi_j)\right|^2dx_jdy_jdz_j
    \leq
    c\varepsilon^{2(\mu_2+1-\sigma_1)},
\end{multline*}
where $\sigma_1=2\sigma(\mu_2-\gamma+3/2)$. There holds a similar inequality with  $\mathbf{w}^l_j$ changed for $\mathbf{w}^r_j$. In view of (\ref{w1p}) and (\ref{w2p}), we obtain
\[
\left\|[\triangle,\Theta] w_{jp};
   V^0_{\gamma,\delta}(G(\varepsilon))\right\|
   \leq c\varepsilon^{\mu_1+\mu_2+1-\sigma_1}.
\]
Finally, using (\ref{w1p}) and (\ref{w2p}) once more, taking into account the estimate
\begin{eqnarray*}
    &&\int_{G(\varepsilon)}(r_j^2+\varepsilon^2)^{\gamma}\left|\Theta(\varepsilon^{-2\sigma}r_j)
   \mathbf{w}_j^l(\varepsilon^{-1}x_j,\varepsilon^{-1}y_j,\varepsilon^{-1}z_j)\right|^2dx_jdy_jdz_j\\
   &=&\varepsilon^{2\gamma+3}\int_{\Omega}(\rho_j^2+1)^{\gamma}\left|\Theta(\varepsilon^{1-2\sigma}\rho_j)
    \mathbf{w}_j^l(\xi_j,\eta_j,\zeta_j)\right|^2d\xi_jd\eta_jd\zeta_j
    \leq
    c\varepsilon^{2\gamma+3},
\end{eqnarray*}
and a similar estimate for $\mathbf{w}_j^r$, we derive
\[\left\|(x,y)\mapsto\Theta(\varepsilon^{-2\sigma}r_j)
w_{jp}(\varepsilon^{-1}x_j,\varepsilon^{-1}y_j,\varepsilon^{-1}z_j);
   V^0_{\gamma,\delta}(G(\varepsilon))\right\|
   \leq c\varepsilon^{\mu_1+\gamma+3/2}.\]
Combining the obtained inequalities, we arrive at  (\ref{estimate vp}).

{\it Step} C\@. This part contains a somewhat modified argument in the proof of Theorem  5.5.1~\cite{MNP}. Let us rewrite the right-hand side of problem (\ref{f0 problem in G(eps1eps2)}) in the form
\begin{align}
\nonumber
f(x,y,z) &
=f_1(x,y,z;\varepsilon)+f_2(x,y,z;\varepsilon)+f_3(x,y,z;\varepsilon)
\\\label{f in G(eps1eps2)}
    & \quad +\varepsilon^{-\gamma-3/2}F_1(\varepsilon^{-1}x_1,\varepsilon^{-1}y_1,\varepsilon^{-1}z_1;\varepsilon_1)+
    \varepsilon^{-\gamma-3/2}F_2(\varepsilon^{-1}x_2,\varepsilon^{-1}y_2,\varepsilon^{-1}z_2;\varepsilon),
\end{align}
where
\begin{eqnarray*}
    &&f_l(x,y,z;\varepsilon) =
  \chi_{l,\varepsilon^{\sigma}}(x,y,z)f(x,y,z),\\
  &&F_j(\xi_j,\eta_j,\zeta_j;\varepsilon) =
   \varepsilon^{\gamma+3/2}\Theta(\varepsilon^{1-\sigma}\rho_j)
   f(x_{O_j}+\varepsilon\xi_j,y_{O_j}+\varepsilon\eta_j,z_{O_j}+\varepsilon\zeta_j);
\end{eqnarray*}
$(x,y,z)$ are arbitrary Cartesian coordinates; $(x_{O_j},y_{O_j},z_{O_j})$ denote the coordinates of the point $O_j$ in the system  $(x,y,z)$; $x_j,y_j,z_j$ were introduced in Section \ref{s2new}.
From the definitions of the norms it follows that
\begin{equation}
\label{estimates fF in G(eps1eps2)}
\|f_1;V^0_{\gamma,\,
\delta}(G_1)\| +\|f_2;V^0_{\gamma}(G_2)\| +\|f_3;V^0_{\gamma,\,
\delta}(G_3)\| +\|F_j;V_{\gamma}^0(\Omega_j)\|\leq
\|f;V^0_{\gamma,\,\delta}(G(\varepsilon))\|.
\end{equation}
We consider solutions   $v_l$ and $w_j$ of the limit problems \[
\begin{aligned}
-(-i\nabla+\textbf{A})^2v\pm Hv+k^2v & =f_2 \, \text{in } G_2,\quad & v & =0 \  \text{on } \partial G_2,\\
\triangle v+k^2v & =f_l \, \text{in } G_l,\quad & v & =0 \  \text{on } \partial G_l,\quad l=1,3,\\
\triangle w & =F_j \, \text{in } \Omega_j,\quad & w & =0 \ 
\text{on } \partial\Omega_j,
\end{aligned}
\]
respectively; besides,  $v_l$ with $l=1,3$ satisfy the intrinsic radiation conditions at infinity, 
whereas  $v_2$ is subject to the condition $(v_2,v_0)_{G_2}=0$. According to Proposition \ref{prop2.1},
\ref{prop2.2}, and ~\ref{prop2.3}, the problems in $G_l$ and 
$\Omega_j$ are uniquely solvable and 
\begin{equation}
\label{estimates vjw in G(eps1eps2)}
\begin{aligned}
\|v_2;V^2_{\gamma}(G_2) \|& \leq
      c_2\|f_2;V^0_{\gamma}(G_2)\|,\\
\|v_l;V^2_{\gamma, \delta, -}(G_l)\| & \leq c_l\|f_l;V^0_{\gamma,\delta}(G_l)\|,\;l=1,3\\
\|w_j;V_{\gamma}^2(\Omega_j)\| & \leq
C_j\|F_j;V_{\gamma}^0(\Omega_j)\|,\;j=1,2,
\end{aligned}
\end{equation}
where  $c_l$ and  $C_j$ are independent of  $\varepsilon$. We set
\begin{align*}
U(x,y,z;\varepsilon)&=\chi_{1,\varepsilon}(x,y,z)v_1(x,y,z;\varepsilon)+
\varepsilon^{-\gamma+3/2}\Theta(r_1)w_1(\varepsilon^{-1}x_1,\varepsilon^{-1}y_1,\varepsilon^{-1}z_1;\varepsilon)\\
&+\chi_{2,\varepsilon}(x,y,z)v_2(x,y,z;\varepsilon)+
\varepsilon^{-\gamma+3/2}\Theta(r_2)w_2(\varepsilon^{-1}x_2,\varepsilon^{-1}y_2, \varepsilon^{-1}z_2;\varepsilon)\\
&+\chi_{3,\varepsilon}(x,y,z)v_3(x,y,z;\varepsilon). 
\end{align*}
The estimates  (\ref{estimates fF in G(eps1eps2)}) and (\ref{estimates vjw in G(eps1eps2)}) lead to \begin{equation}
\label{estimate U with f in G(eps1eps2)}
    \|U;V^2_{\gamma,\, \delta, -}(G(\varepsilon))\|\leq
    c\|f;V^0_{\gamma,\delta}(G(\varepsilon))\|
\end{equation}
with constant  $c$ independent of $\varepsilon$. Denote the operator  $f\mapsto U$ by $R_{\varepsilon}$. Arguing as in the proof of  \cite[Theorem 5.5.1]{MNP}, we obtain
$(-(-i\nabla+\textbf{A})^2\pm H+k^2)R_{\varepsilon}=I+S_{\varepsilon}$,
where $S_{\varepsilon}$ is an operator with small norm in  $V^0_{\gamma,\delta}(G(\varepsilon))$.

{\it Step} D\@. Recall that the operator $S_{\varepsilon}$ is defined on the subspace 
$V_{\gamma,\,\delta}^{0,\perp}(G(\varepsilon))$. We need that the range of 
$S_{\varepsilon}$ would also be in  $V_{\gamma,\,\delta}^{0,\perp}(G(\varepsilon))$.
To this end we change  $R_{\varepsilon}$ for $\widetilde{R}_{\varepsilon}:f\mapsto U(f)+a(f)u_p$, where $u_p$ was constructed at step ~{\bf A}\@, $a(f)$ being a constant. Then 
$(-(-i\nabla+\textbf{A})^2\pm H+k^2)\widetilde{R}_{\varepsilon}=I+\widetilde{S}_{\varepsilon}$ with
$\widetilde{S}_{\varepsilon}=S_{\varepsilon}+a(\cdot)(-(-i\nabla+\textbf{A})^2\pm H+k^2)u_p$. The condition
$(\chi_{2,\varepsilon^{\sigma}}\widetilde{S}_{\varepsilon}f,v_0)_{G_2}=0$
with $k=k_0$ implies that
$a(f)=-(\chi_{2,\varepsilon^{\sigma}}S_{\varepsilon}f,v_0)_{G_2}/
(\chi_{2,\varepsilon^{\sigma}}(-(-i\nabla+\textbf{A})^2\pm H+k^2_0)u_p,v_0)_{G_2}$.
We show that  $\|\widetilde{S}_{\varepsilon}\|\leq
c\|S_{\varepsilon}\|$, where $c$ is independent of
$\varepsilon$ and $k$. We have \[
\|\widetilde{S}_{\varepsilon}f\|\leq
\|S_{\varepsilon}f\|+
|a(f)|\,\|(-(-i\nabla+\textbf{A})^2\pm H+k^2)u_p\|.
\]
The estimate (\ref{estimate vp}) (with $\gamma>\mu_1-1/2$ and
$\mu_2-\mu_1>\sigma_1$), the formula for $k_p$, and the condition
$k^2-k^2_0=O\bigl(\varepsilon^{2\mu_1+1}\bigr)$
lead to the inequality
\begin{eqnarray*}
&&\|(-(-i\nabla+\textbf{A})^2\pm H+k^2)u_p;V_{\gamma,\delta}^0\| \\
&&\leq
|k^2-k^2_p|\,\|u_p;V_{\gamma,\delta}^0\|+
\|(-(-i\nabla+\textbf{A})^2\pm H+k_p^2)u_p;V_{\gamma, \delta}^0\|  \leq
c \varepsilon^{2\mu_1+1}.
\end{eqnarray*}
The supports of the functions $(-(-i\nabla+\textbf{A})^2\pm H+k^2_p)u_p$ and
$\chi_{2,\varepsilon^{\sigma}}$ are disjoint, so   \[
|(\chi_{2,\varepsilon^{\sigma}}(-(-i\nabla+\textbf{A})^2\pm H+k^2_0)u_p,v_0)_{G_2}|=|(k^2_0-k^2_p)(u_p,v_0)_{G_2}|\geq
c \varepsilon^{2\mu_1+1}.
\]
Further,  $\gamma-1<\mu_1+1/2$, therefore 
\begin{align*}
|(\chi_{2,\varepsilon^{\sigma}}S_{\varepsilon}f,v_0)_{G_2}| \leq
\|S_{\varepsilon}f;V_{\gamma,\delta}^0(G(\varepsilon))\|
\,\|v_0;V_{-\gamma}^0(G_2)\|\leq
c\|S_{\varepsilon}f;V_{\gamma,\delta}^0(G(\varepsilon))\|.
\end{align*}
Hence\[
|a(f)|\leq
c\varepsilon^{-2\mu_1-1}\|S_{\varepsilon}f;V_{\gamma,\delta}^0(G(\varepsilon))\|
\]
and $\|\widetilde{S}_{\varepsilon}f\|\leq
c\|S_{\varepsilon}f\|$. It follows that the operator 
$I+\widetilde{S}_{\varepsilon}$ in
$V_{\gamma,\delta}^{0,\perp}(G(\varepsilon))$
invertible as well as the operator of problem
(\ref{f0 problem in G(eps1eps2)}):\[
A_{\varepsilon}:u\mapsto-(-i\nabla+\textbf{A})^2u\pm H
u+k^2u:\mbox{\emph{\r{V}}}{}^{2,\perp}_{\gamma, \delta,
-} (G(\varepsilon))\mapsto
V^{0,\perp}_{\gamma,\delta}(G(\varepsilon));
\]
here $\mbox{\emph{\r{V}}}{}^{2,\perp}_{\gamma, \delta, -}
(G(\varepsilon))$ stands for the space of functions in
$V^2_{\gamma, \delta, -}
(G(\varepsilon))$ that vanish at  $\partial
G(\varepsilon)$ and are sent by the operator 
$-(-i\nabla+\textbf{A})^2\pm H+k^2$ to $V^{0,\perp}_{\gamma,\delta}$. The inverse operator
$A_{\varepsilon}^{-1}=\widetilde{R}_{\varepsilon}(I+\widetilde{S}_{\varepsilon})^{-1}$
has been bounded uniformly with respect to 
$\varepsilon$ and $k$. Therefore,  
(\ref{estimate probl G(eps1eps2)}) holds with a constant  $c$ independent of  
$\varepsilon$ and $k$.\qquad
\end{proof}
We consider the solution $u_1$ to the homogeneous problem (\ref{scalar Pauli equation}) satisfying 
\[u_1(x,y,z) = \begin{cases}
    U_1^+(x,y,z)+s_{11}\,U^-_1(x,y,z)+O(\exp{(\delta x)}),  & x\rightarrow -\infty, \\
               s_{12}\,U^-_{2}(x,y,z)+O(\exp{(-\delta x)}),  & x\rightarrow +\infty. \end{cases}\]
Let $s_{11}$ and $s_{12}$ be the entries of the scattering matrix determined by this solution. Denote by  $\widetilde{u}_{1,\sigma}$ the function given by (\ref{asympt of wave function}) changing  $\Theta(r_j)$  
for $\Theta(\varepsilon_j^{-2\sigma}r_j)$ and dropping the remainder $R$, while $\widetilde{s}_{11}$, and  $\widetilde{s}_{12}$  stand for the quantities defined in (\ref{s11}) and (\ref{s12}).

\begin{thm}\label{thm5.4}
Let the assumptions of Proposition \/ {\rm\ref{prop5.3}} be fulfilled. Then the inequality
\[
|s_{11}-\widetilde{s}_{11}| +|s_{12}-\widetilde{s}_{12}|\leq
c|\widetilde{s}_{12}|\varepsilon^\tau
\]
holds with constant  $c$ independent of $\varepsilon,k$; $\tau=\min\{2-\delta,\mu_2-\mu_1\}$ and with arbitrarily small positive $\delta$.\end{thm}

\begin{proof} 
The difference  $R=u_1-\widetilde{u}_{1,\sigma}$ belongs to  $V^2_{\gamma,\,\delta,
-}(G(\varepsilon))$, whereas 
$f_1:=(-(-i\nabla+\textbf{A})^2\pm H+k^2)(u_1-\widetilde{u}_{1,\sigma})$ is in
$V^{0,\perp}_{\gamma,\,\delta}(G(\varepsilon))$.  By Proposition ~\ref{prop5.3},
\begin{equation}\label{estimate R}
\|R;V^2_{\gamma, \,\delta,
-}(G(\varepsilon))\|\leq
c\,\|f_1;V^0_{\gamma,\delta}(G(\varepsilon))\|.
\end{equation}
We show that
\begin{equation}\label{estimate f1}
\|f_1;V^0_{\gamma,\,\delta}(G(\varepsilon))\|\leq
c|\widetilde{s}_{12}|(\varepsilon^{\gamma-\mu_1+1/2}+\varepsilon^{\mu_2-\mu_1-\sigma_1}),
\end{equation}
where $\sigma_1=2\sigma(\mu_2-\gamma+3/2)$. Then the desired estimate will follow from the last two  inequalities with $\gamma=\mu_1+3/2-\delta$ and $\sigma_1=\delta$. 

Arguing as in the step B of the proof of Proposition ~\ref{prop5.3} we obtain
\begin{eqnarray*}
\|f_1;V^0_{\gamma,\,\delta}(G(\varepsilon))\|&\leq& c(\varepsilon^{\gamma+3/2}+\varepsilon^{\mu_2+1-\sigma_1})\\
&\times&\max_{j=1,2}(|a_j^-(\varepsilon)|\varepsilon^{-\mu_1-1}  +|a_j^+(\varepsilon)|\varepsilon^{\mu_1}+|b_j^-(\varepsilon)|\varepsilon^{-\mu_1-1}  +|b_j^+(\varepsilon)|\varepsilon^{\mu_1}).
\end{eqnarray*}
From (\ref{ab conditions}) it follows that 
\[(|a_j^-(\varepsilon)|\varepsilon^{-\mu_1-1}  +|a_j^+(\varepsilon)|\varepsilon^{\mu_1})\leq c(|b_j^-(\varepsilon)|\varepsilon^{-\mu_1-1}  +|b_j^+(\varepsilon)|\varepsilon^{\mu_1}).\]
Taking into account (\ref{b1+-}) and (\ref{ab2+-}) for  $b_j^{\pm}$ and also (\ref{C1C2}) and (\ref{gammadelta}), we derive
\[|b_j^-(\varepsilon)|\varepsilon^{-\mu_1-1}  +|b_j^+(\varepsilon)|\varepsilon^{\mu_1} \leq c \varepsilon^{-\mu_1-1}|\widetilde{s}_{12}(\varepsilon)|.\]
Combining the obtained estimates, we arrive at (\ref{estimate f1}).\qquad\end{proof}
Theorem \ref{thm5.4} together with   (\ref{T}) and (\ref{quality factor}) lead to the following assertion. We return here to the detailed notations introduced in Sections \ref{s2new} - \ref{s4new}. 

\begin{thm}\label{thm6.3}

For \/ $|k^2-k_{r,\pm}^2|=O(\varepsilon^{2\mu_1+1})$ there hold the asymptotic representations 
\begin{align*}
T^\pm(k,\varepsilon)&=\frac{1}{\dfrac{1}{4}\left(\dfrac{|b_1^\pm|}{|b_2^\pm|}+\dfrac{|b_2^\pm|}{|b_1^\pm|}\right)^2+P_\pm^2\left(\dfrac{k^2-k^2_{r,\pm}}{\varepsilon^{4\mu_1+2}}\right)^2}(1+O(\varepsilon^\tau)),\\
k^2_{r,\pm}&=k^2_{0,\pm}-\alpha
(|b_1^\pm|^2+|b_2^\pm|^2)\varepsilon^{2\mu_1+1}+
O\bigl(\varepsilon^{2\mu_1+1+\tau}\bigr),\\
\Upsilon^\pm(\varepsilon)&=
\left(\dfrac{|b_1^\pm|}{|b_2^\pm|}+\dfrac{|b_2^\pm|}{|b_1^\pm|}\right)P^{-1}_\pm\varepsilon^{4\mu_1+2}\bigl(1+O(\varepsilon^\tau)\bigr),
\end{align*}
where $\Upsilon^\pm(\varepsilon)$ is the width of the resonant peak at its half-height (the so-called resonant quality factor), $P_\pm=(2|b_1^\pm||b_2^\pm|\beta^2|A(k_0)|^2)^{-1}$, and $\tau=\min\{2-\delta,\mu_2-\mu_1\}$, $\delta$ being an arbitrary small positive number.\end{thm}

\end{document}